\documentclass[11pt]{article}%
\usepackage{amsfonts}
\usepackage{amsmath}
\usepackage{amssymb}
\usepackage{graphicx}
\usepackage{bbm}%
\providecommand{\U}[1]{\protect \rule{.1in}{.1in}}
\providecommand{\U}[1]{\protect \rule{.1in}{.1in}}
\providecommand{\U}[1]{\protect \rule{.1in}{.1in}}
\newtheorem{theorem}{Theorem}

\newtheorem{lemma}[theorem]{Lemma}

\newenvironment{proof}[1][Proof]{\noindent \textbf{#1.} }{\  \rule{0.5em}{0.5em}}
\def \doublespace {\openup 2.0\jot}
\doublespace
\begin{document}

\title{Semiparametric Regression Models for Explanatory Variables with Missing Data
due to Detection Limit}
\author{Zhang, J.$^{1\dag}$, Shao, L.$^{1\dag}$, Yang, K.$^{1}$, Quach, N.E.$^{1}$, Tu, S.$^{1}%
$, Chen, R.$^{2}$
\and Wu, T.$^{1}$, Liu, J.$^{3}$, Tu, J.$^{4}$, Suarez-Lopez, J.R.$^{5}$, Zhang,
X.$^{1}$, Lin, T.$^{6\ast\ddag}$ and Tu, X.M.$^{1\ast}$}
\maketitle

\begin{center}
	$^{1}$Division of Biostatistics and Bioinformatics\\UCSD\ Herbert Wertheim School of Public Health and Human Longevity Science\\La Jolla, CA 92093\\$^{2}$Division of Biostatistics\\Feinberg School of Medicine, Northwestern University, Chicago, IL 60611\\$^{3}$Department of Biostatistics\\Vanderbilt University Medical Center, Nashville, TN 37232\\$^{4}$Department of Orthopedics, \\Emory Health Care, Emory University, Atlanta, GA 30329\\$^{5}$UCSD\ Herbert Wertheim School of Public Health and Human Longevity Science\\La Jolla, CA 92093\\$^{6}$Department of Biostatistics\\University of Florida, Gainesville, FL 32608
\end{center}
$^{\dag}$These authors contributed equally to this work.\\
$^{\ast}$These authors jointly supervised this work.\\
$^{\ddag}$Corresponding author: Tuo Lin, \texttt{tuolin@ufl.edu}

\begin{center}
{\large Summary}
\end{center}

Detection limit (DL) has become an increasingly ubiquitous issue in statistical analyses of biomedical studies, such as cytokine, metabolite and protein analysis.  In regression analysis, if an explanatory variable is left-censored due to concentrations below the DL, one may limit analyses to observed data.  In many studies, additional, or surrogate, variables are available to model, and incorporating such auxiliary modeling information into the regression model can improve statistical power.  Although methods have been developed along this line, almost all are limited to parametric models for both the regression and left-censored explanatory variable. While some recent work has considered semiparametric regression for the censored DL-effected explanatory variable, the regression of primary interest is still left parametric, which not only makes it prone for biased estimates, but also suffers from high computational cost and inefficiency due to maximizing an extremely complex likelihood function and bootstrap inference. In this paper, we propose a new approach by considering semiparametric generalized linear models (SPGLM) for the primary regression and parametric or semiparametric models for DL-effected explanatory variable. The semiparametric and semiparametric combination provides the most robust inference, while the semiparametric and parametric case enables more efficient inference. The proposed approach is also much easier to implement and allows for leveraging sample splitting and cross fitting (SSCF) to improve computational efficiency in variance estimation. In particular, our approach improves computational efficiency over bootstrap by 450 times. We use simulated and real study data to illustrate the approach.   


\section{Introduction\label{sec1}}

Detection limit (DL) due to limited precision in measurement of an outcome has
become increasingly a ubiquitous issue in clinical laboratory, drug development,
environmental and molecular biology. \ Specifically, in biomarker studies such
as cytokine, metabolite and protein analysis, the limitations of
innovative assay technologies generate data with high amount of missing data
due to DL \cite{wong2008reproducibility,hrydziuszko2012missing}. \ Although a DL threshold can be either a flooring or ceiling
value, DL due to flooring effects is the most common in modern biomedical
research and as such has been the focus in recent literatures
\cite{richardson2003effects,schisterman2006limitations,pickering2017rapid,klymus2020reporting}%
. \ If $x$ denotes a continuous outcome subject to DL and $\delta$ denotes the threshold, then $x$ is observed only when $x > \delta$. \ Note that in some cases, such as SomaScan proteomics assay data, although
measurements may still be available for $x\leq \delta$, they are typically
unreliable and should not be used in statistical analysis
\cite{candia2022assessment,armbruster2008limit}. \ 

If $x$ is used as a response, or dependent variable, parameter
estimates of explanatory, or independent, variables are generally biased if missing
data due to DL\ is completely ignored. Various statistical approaches have
been developed to address this challenge. \ For example, Lubin et al.
\cite{lubin2004epidemiologic} summarized two common approaches, the Tobit model \cite{tobin1958estimation} and multiple imputation
method \cite{lee2010multiple,little2019statistical}. \ To minimize the restrictive parametric model assumption, Zhang et
al. \cite{zhang2009nonparametric} proposed nonparametric methods for data with
high rates of DL and relative small sample size, but these methods only work
for the two-sample problem. \ Dutta et al. \cite{dutta2021semiparametric}
applied the pseudo-value technique and used semiparametric models for
inference. Tian et al. \cite{tian2024addressing}
developed a rank based cumulative
probability model for robust inference. \ 

When used as an explanatory variable, analyses based on observed data yield consistent estimates. \ However, potential
efficiency loss due to missing data creates high incentives to go beyond the
observed data. \ For example, in one of our recent studies investigating
effects of secondary exposures to pesticides on neurocogntive outcomes in
agricultural residents (see Section \ref{sec5} for details)
\cite{suarez2012lower}, missing data due to concentrations below the DL\ for most of the metabolites
exceed 30\%. \ By incorporating auxiliary information for missing $x$, we may significantly improve
statistical power. \ 

As in the case where $x$ is used as a response, available approaches for handling $x$ as an explanatory variable are largely
based on parametric models, using likelihood or multiple imputation techniques by positing
an additional parametric model for $x$
\cite{ye2024joint,arunajadai2012handling,bernhardt2015statistical,harel2014use}.
\ Recently, Kong and Nan \cite{kong2016semiparametric} and Chen et al.
\cite{chen2022semiparametric} considered an auxiliary semiparametric model for $x$ for more robust inference. \ However, with the primary regression continuing to follow the exponential family of distributions, biased estimates will arise if the parametric distribution is incorrectly specified. \ Additionally, solving the complex score equations is quite computationally intensive and algebraically complex. \ The bootstrap variance estimation further increases the computational burden, making it difficult to use in practice. \ For example, when applied to our real study data in Section \ref{sec5}, this approach takes approximately 7.5 hours. 

In this paper, we leverage the semiparametric theory to further
extend the primary regression to semiparametric models, while allowing for both
parametric and semiparametric models for $x$ to provide robust and efficient
inference. \ The proposed approach not only provides valid inference for virtually all data distribution, but also is much simpler to implement and much more computationally efficient. \ When applied to the same real study data in Section \ref{sec5}, it takes less than one minute to compute all estimates for inference. \ The rest of the paper is structured as follows. \ Section
\ref{sec2} introduces a two-component modeling approach, followed by inference
for model parameters. \ Section \ref{sec4} includes
simulation studies to assess the consistency, robustness and efficiency gain
of the proposed approach, while Section \ref{sec5} contains a real study
application. \ Section \ref{sec6} highlights findings and discusses
limitations and further extensions of the approach. \ 

\section{A Two-Component Modeling Approach\label{sec2}}

\subsection{Model Setup}
Let $y_{i}$ denote a response, $x_{i}$ a continuous explanatory variable
subject to DL with a threshold $\delta$, and $\mathbf{u}_{i}$ a vector of
remaining explanatory variables not subject to DL from the $i$th subject
($1\leq i\leq n$). \ Let $\mathbf{z}_{i}=\left(  z_{i1},z_{i2},\ldots
,z_{im}\right)  ^{\top}$ denote a vector of surrogate variables for modeling
missing $x_{i}$ due to DL, which may overlap with $\mathbf{u}_{i}$. \ The
surrogates provide information to model DL-effected missing in $x$ so we can incorporate such information to improve efficiency for
the regression of interest.
\ We assume no other type of missing in any of the variables. \ We denote the
observed data as: $\left \{  y_{i},\mathbf{w}_{i}\right \}  $ with
$\mathbf{w}_{i}=\left \{  x_{i}^{obs},\mathbf{u}_{i},\mathbf{z}_{i}\right \}  $
$\left(  1\leq i\leq n\right)  $, where $x_{i}^{obs} = \{x_i>\delta, \mathbbm{1}(x_i<\delta)\}$, and $\mathbbm{1}(\cdot)$ denotes an indicator function. \ 

Our primary interest lies in the relationship of $y_{i}$ with $\mathbf{u}_{i}$
and $x_{i}$ in a semiparametric generalized linear model (SPGLM), or
restricted moment model \cite{tsiatis2006semiparametric,tang2023applied}:\
\begin{equation}
E\left(  y_{i}\mid \mathbf{u}_{i},x_{i}\right)  =h\left(  x_{i},\mathbf{u}%
_{i};\boldsymbol{\beta}\right)  ,\quad1\leq i\leq n \label{eqn10}%
\end{equation}
where $h\left(  \cdot \right)  $ denotes the inverse of an appropriate link
function and $\boldsymbol{\beta}=\left(  \beta_{0},\beta_{1},\boldsymbol{\beta
}_{2}^{\top}\right)  ^{\top}$ denotes the vector of parameters. Unlike its
parametric counterpart, SPGLM imposes no parametric distribution for $y$ given $\mathbf{u}$ and $x$ and thus
provides valid inference for virtually all data distributions
\cite{tsiatis2006semiparametric,tang2023applied}. \ If ignoring the missing
$x_{i}$ due to DL, we can readily fit the model in (\ref{eqn10}) based on the
subsample with observed $x_{i}$ $\left(  >\delta \right)  $. \ However, this
approach may be inefficient as noted earlier, especially with a high rate of
missing $x_{i}$. \ 

To utilize the surrogates $\mathbf{z}_{i}$ for information about the missing
$x_{i}$ due to DL, we consider a second, or auxiliary, linear regression
model:\
\begin{equation}
x_{i}=\gamma_{0}+\boldsymbol{\gamma}_{1}^{\top}\mathbf{z}_{i}+\epsilon
_{i},\quad1\leq i\leq n \label{eqn50}%
\end{equation}
where $\epsilon_{i}$ follows either a parametric or non-parametric
distribution. \ For parametric models, the distribution of $\epsilon_{i}$ is
determined by a finite-dimensional vector of parameters $\xi$. \ For example,
under a normal $\epsilon_{i}\sim N\left(  0,\sigma_{x}^{2}\right)  $,
(\ref{eqn50})\ is parameterized by $\boldsymbol{\eta}=\left(
\boldsymbol{\gamma}^{\top},\xi \right)  ^{\top}$, with $\xi=\sigma_{x}^{2}$ and $\boldsymbol{\gamma}=\left(
\gamma_{0},\boldsymbol{\gamma}_{1}^{\top}\right)  ^{\top}$. \ For
non-parametric $\epsilon_{i}$, $\xi \left(  \cdot \right)  $ is
infinite-dimensional denoted $\boldsymbol{\eta}=\left(  \boldsymbol{\gamma}%
^{\top},\xi \left(  \cdot \right)  \right)  ^{\top}$. \ For example, in Kong and
Nan \cite{kong2016semiparametric}, the left-censored $x_{i}$ is transformed to
a right-censored $t_{i}$ (using a monotone decreasing function, e.g., $x = T(t)=\exp(-t)$ and $T(t)=-t$), which is modeled
using a semiparametric linear model:
\begin{equation}
t_{i}=\zeta_{0}+\boldsymbol{\zeta}_{1}^{\top}\mathbf{z}_{i}+\varepsilon
_{i},\quad1\leq i\leq n \label{eqn54}%
\end{equation}
where $\varepsilon_{i}$ is non-parametric with an infinite-dimensional
$\xi \left(  t\right)  $ $\left(  t\geq0\right)  $. \ For notational
simplicity, we use $\boldsymbol{\gamma}$ \ to refer to the regression
coefficients, and $\epsilon_{i}$ to the error term for both (\ref{eqn50}) and
(\ref{eqn54}), with finite- and
infinite-dimensional $\boldsymbol{\eta}$ for the respective parametric and semiparametric models for $x_{i}%
$. \ 

To incorporate (\ref{eqn50}) into (\ref{eqn10}), we assume for the observed $x_i$:
\begin{equation}
E\left(  y_{i}\mid \mathbf{u}_{i},x_{i},\mathbf{z}_{i}\right)  =E\left(
y_{i}\mid \mathbf{u}_{i},x_{i}\right)  =h\left(  x_{i},\mathbf{u}%
_{i};\boldsymbol{\beta}\right)  \label{eqn60}%
\end{equation}
The condition in (\ref{eqn60})\ is akin to the \textquotedblleft
surrogacy\textquotedblright \ assumption in measurement error, or errors in
variables, models \cite{kowalski2002generalized,richardson2003effects}. \ It
simply states that given $x_{i}$ the surrogates $\mathbf{z}_{i}$ provide no
additional information for the primary regression model (\ref{eqn10}). \ 

With the above assumptions, we now have a two-component model for the observed
data $\left \{  y_{i},\mathbf{w}_{i}\right \}  $:%
\begin{align}
E\left(  y_{i}\mid \mathbf{w}_{i}\right)   &  =g\left(  \mathbf{w}%
_{i};\boldsymbol{\theta}\right)  ,\quad \boldsymbol{\theta}=\left(
\boldsymbol{\beta}^{\top},\boldsymbol{\eta}^{\top}\right)  ^{\top},\quad1\leq
i\leq n,\label{eqn70}\\
g\left(  \mathbf{w}_{i};\boldsymbol{\theta}\right)   &  =\left \{
\begin{array}
[c]{ll}%
h\left(  x_{i},\mathbf{u}_{i};\boldsymbol{\beta}\right)  & \text{if }%
x_{i}>\delta \\
\widetilde{h}\left(  \mathbf{z}_{i},\mathbf{u}_{i};\boldsymbol{\beta
},\boldsymbol{\eta}\right)  & \text{if }x_{i}\leq \delta
\end{array}
\right. \nonumber \\
x_{i}  &  =\gamma_{0}+\boldsymbol{\gamma}_{1}^{\top}\mathbf{z}_{i}%
+\epsilon_{i},\quad1\leq i\leq n\nonumber
\end{align}
where
\begin{align*}
\widetilde{h}\left(  \mathbf{z}_{i},\mathbf{u}_{i};\boldsymbol{\beta
},\boldsymbol{\eta}\right)   &  =E\left[  E\left(  y_{i}\mid \mathbf{z}%
_{i},\mathbf{u}_{i},x_{i}\right)  \mid \mathbf{z}_{i},\mathbf{u}_{i},x_{i}%
\leq \delta;\boldsymbol{\eta}\right] \\
&  =E\left[  h\left(  x_{i},\mathbf{u}_{i};\boldsymbol{\beta}\right)
\mid \mathbf{z}_{i},\mathbf{u}_{i},x_{i}\leq \delta;\boldsymbol{\eta}\right]
\end{align*}
We refer to the setting with a semiparametric primary component and a parametric auxiliary component as semi-para; the other configuration semi-semi is defined analogously. \ Unlike the model in (\ref{eqn10}), $\boldsymbol{\theta}=\left(
\boldsymbol{\beta}^{\top},\boldsymbol{\eta}^{\top}\right)  ^{\top}$ is the
parameter vector for the two-component model, with $\boldsymbol{\beta}$\ being
the parameters of interest for the primary model in
(\ref{eqn10}) and $\boldsymbol{\eta}$ the nuisance parameters for the
parametric or semiparametric linear model for $x_{i}$. \

The conditional mean $g\left(  \mathbf{w}_{i};\boldsymbol{\theta}\right)  $ in
(\ref{eqn70}) generally cannot be expressed in closed form for $x_{i}%
\leq \delta$. \ For a normal $\epsilon_{i}\sim N\left(  0,\sigma_{x}%
^{2}\right)  $, $g\left(  \mathbf{w}_{i};\boldsymbol{\theta}\right)  $ is in
closed form (see Appendix for details):
\begin{align}
g\left(  \mathbf{w}_{i};\boldsymbol{\theta}\right)   &  =\left \{
\begin{array}
[c]{ll}%
\beta_{0}+\beta_{1}x_{i}+\boldsymbol{\beta}_{2}^{\top}\mathbf{u}_{i} &
\text{if }x_{i}>\delta \\
\beta_{0}+\beta_{1}E\left(  x_{i}\mid \mathbf{z}_{i},x_{i}\leq \delta \right)
+\boldsymbol{\beta}_{2}^{\top}\mathbf{u}_{i} & \text{if }x_{i}\leq \delta
\end{array}
\right. \label{eqn80}\\
E\left(  x_{i}\mid \mathbf{z}_{i},x_{i}\leq \delta \right)   &  =\gamma
_{0}+\boldsymbol{\gamma}_{1}^{\top}\mathbf{z}_{i}-\sigma_{x}^{2}\frac
{\phi_{\mu_{i},\sigma_{x}^{2}}\left(  \delta \right)  }{\Phi_{\mu_{i}%
,\sigma_{x}^{2}}\left(  \delta \right)  }\nonumber
\end{align}
where $\phi_{\mu_{i},\sigma_{x}^{2}}(x)$ and $\Phi_{\mu_{i},\sigma_{x}^{2}%
}(x)$ denote the PDF and CDF of $x_{i}$ given $\mathbf{z}_{i}$ with mean
$\mu_{i}=\gamma_{0}+\boldsymbol{\gamma}_{1}^{\top}\mathbf{z}_{i}$ and variance
$\sigma_{x}^{2}$ evaluated at $x$. \ 

Variability of $x_{i}$ (given $\mathbf{z}_{i}$) in the auxiliary model in (\ref{eqn50})
(or (\ref{eqn54})) relative to that of $y_{i}$ (given $\mathbf{u}_{i},x_{i}$)
in the primary model plays a critical role for efficiency gains. \ For example, if
$x_{i}$ is much more variable than $y_{i}$, it may offset efficiency gain,
making the estimator of $\boldsymbol{\beta}$ from (\ref{eqn70}) less efficient
than that from (\ref{eqn10}) based on the observed data $\left(
x_{i}>\delta \right)  $ only. \ We investigate this trade-off in Section \ref{sec4}
using simulation studies. \ 

\subsection{Inference\label{sec3}}

Our primary goal is estimating $\boldsymbol{\beta}$. \ We start by assuming
that $\boldsymbol{\eta}$ is known, in which case $\boldsymbol{\eta=\eta}_{0}$
with $\boldsymbol{\eta}_{0}$ being a vector of known constants. \ Then
$\boldsymbol{\beta}$ is the only parameters. \ Let%
\begin{equation}
g_{i}=g\left(  \mathbf{w}_{i};\boldsymbol{\beta}\right)  ,\quad D_{i}%
=\frac{\partial}{\partial \boldsymbol{\beta}^{\top}}g_{i},\quad S_{i}%
=y_{i}-g_{i},\quad V_{i}=Var_{w}\left(  y_{i}\mid \mathbf{w}_{i}\right)
\label{eqn220}%
\end{equation}
where $Var_{w}\left(  y_{i}\mid \mathbf{w}_{i}\right)  $ denotes the working
variance of $y_{i}$ given $\mathbf{w}_{i}$. \ All the terms in (\ref{eqn220})
are well-defined, except for $V_{i}$, since only the conditional mean of
$y_{i}$ given $\mathbf{w}_{i}$ is specified in (\ref{eqn70}). \ In practice,
we specify a working variance $Var_{w}\left(  y_{i}\mid \mathbf{w}_{i}\right)
$, which needs not be the same as the true variance $Var\left(
y_{i}\mid \mathbf{w}_{i}\right)  $. \ For example, one may use $Var_{w}\left(
y_{i}\mid \mathbf{w}_{i}\right)  =\sigma^{2}$ for continuous and $Var_{w}%
\left(  y_{i}\mid \mathbf{w}_{i}\right)  =g_{i}\left(  1-g_{i}\right)  $ for
binary $y_{i}$. \ 

We estimate $\boldsymbol{\beta}$ using the generalized estimating equations
(GEE):
\begin{equation}
U_{n}\left(  \boldsymbol{\beta}\right)  =\sum_{i=1}^{n}U_{ni}\left(
\boldsymbol{\beta}\right)  =\sum_{i=1}^{n}D_{i}V_{i}^{-1}S_{i}=\mathbf{0.}
\label{eqn260}%
\end{equation}
Under mild regularity conditions, the GEE\ estimator $\widehat
{\boldsymbol{\beta}}$ obtained by solving the equations above is consistent
and asymptotically normal:%

\begin{align}
\sqrt{n}\left(  \widehat{\boldsymbol{\beta}}-\boldsymbol{\beta}_{0}\right)
&  \rightarrow_{d}N(0,\Sigma_{\beta}),\label{eqn264}\\
\Sigma_{\beta}  &  =B^{-1}\Sigma_{U}B^{-\top},\quad B=E\left(  D_{i}^{\top
}V_{i}^{-1}D_{i}\right)  ,\quad \Sigma_{U}=E\left(  D_{i}^{\top}V_{i}^{-1}%
S_{i}S_{i}^{\top}V_{i}^{-1}D_{i}\right) \nonumber
\end{align}
where $\boldsymbol{\beta}_{0}$ denotes the true value of $\boldsymbol{\beta}$.
\ The asymptotic variance $\Sigma_{\beta}$ is the variance of the efficient
influence function (IF) for $\widehat{\boldsymbol{\beta}}$,
$\boldsymbol{\varphi}_{i}\left(  y_{i},\mathbf{w}_{i};\boldsymbol{\beta}%
_{0}\right)  =B^{-1}D_{i}V_{i}^{-1}S_{i}$ \cite{tsiatis2006semiparametric}.
\ A consistent estimator $\widehat{\Sigma}_{\beta}$ of 
$\Sigma_{\beta}$ is obtained by substituting the expectations in
(\ref{eqn264}) with their respective sample counterparts:
\begin{align}
\widehat{B}  &  =\frac{1}{n}\sum_{i=1}^{n}\widehat{D}_{i}^{\top}\widehat
{V}_{i}^{-1}\widehat{D}_{i},\quad \widehat{\Sigma}_{U}=\frac{1}{n}\sum
_{i=1}^{n}\widehat{D}_{i}^{\top}\widehat{V}_{i}^{-1}\widehat{S}_{i}\widehat
{S}_{i}^{\top}\widehat{V}_{i}^{-1}\widehat{D}_{i}\label{eqn268}\\
\widehat{\Sigma}_{\beta}  &  =\frac{1}{n}\sum_{i=1}^{n}\boldsymbol{\varphi
}_{i}\left(  y_{i},\mathbf{w}_{i};\widehat{\boldsymbol{\beta}}\right)
\boldsymbol{\varphi}_{i}^{\top}\left(  y_{i},\mathbf{w}_{i};\widehat
{\boldsymbol{\beta}}\right)  =\widehat{B}^{-1}\widehat{\Sigma}_{U}\widehat
{B}^{-\top},\nonumber
\end{align}
where $\widehat{A}_{i}$ denotes $A_{i}$ with $\boldsymbol{\beta}$ \ replaced
by $\widehat{\boldsymbol{\beta}}$. \ The GEE estimator $\widehat
{\boldsymbol{\beta}}$ is also asymptotically efficient
\cite{tsiatis2006semiparametric}. \ 

In most applications, $\boldsymbol{\eta}$ is {\normalsize unknown. \ }If
$\boldsymbol{\eta}$ is finite dimensional, we readily estimate
$\boldsymbol{\eta}$ using maximum likelihood. \ Let $f\left(  x_{i}%
\mid \mathbf{z}_{i};\boldsymbol{\eta}\right)  $ denote the density of $x_{i}$
given $\mathbf{z}_{i}$. \ Since the censoring value $\delta$ is the same for
all subjects, we may view the observed $x_{i}$ as sampled from the (left)
truncated version of $f\left(  x_{i}\mid \mathbf{z}_{i};\boldsymbol{\eta
}\right)  $ and compute the MLE\ of $\boldsymbol{\eta}$ from the truncated
likelihood and associated score given by:
\begin{align*}
l_{n_{o}}\left(  \boldsymbol{\eta}\right)   &  =\sum_{x_{i}>\delta}l{_{i}%
}\left(  \boldsymbol{\eta}\right)  ,\quad Q_{n_{o}}\left(  \boldsymbol{\eta
}\right)  =\sum_{x_{i}>\delta}Q{_{n_{o}i}}\left(  \boldsymbol{\eta}\right) \\
l_{n_{o}i}\left(  \boldsymbol{\eta}\right)   &  =\log \left(  f\left(
x_{i}\mid \mathbf{z}_{i};\boldsymbol{\eta}\right)  \right)  -\log \left(
\int_{\delta}^{\infty}f\left(  x\mid \mathbf{z}_{i};\boldsymbol{\eta}\right)
dx\right)  ,\quad Q{_{n_{o}i}}\left(  \boldsymbol{\eta}\right)  =\frac
{\partial}{\partial \boldsymbol{\eta}^{\top}}l{_{n_{o}i}}\left(
\boldsymbol{\eta}\right)
\end{align*}
where $n_{o}$ denotes the sample size of the observed $x_{i}>\delta$. \ The
MLE $\widehat{\boldsymbol{\eta}}$ from solving the score equations $Q_{n_{o}%
}\left(  \widehat{\boldsymbol{\eta}}\right)  =\mathbf{0}$ is consistent and
asymptotically normal:
\[
\sqrt{n}\left(  \widehat{\boldsymbol{\eta}}-\boldsymbol{\eta}_{0}\right)
=\sqrt{n_{0}}\frac{1}{\sqrt{\frac{n_{0}}{n}}}\left(  \widehat{\boldsymbol{\eta
}}-\boldsymbol{\eta}\right)  \rightarrow_{p}N(\mathbf{0},\frac{1}{p}%
I^{-1}),\quad n\rightarrow \infty,
\]
where $I=E\left[  -\frac{\partial}{\partial \boldsymbol{\eta}}Q{_{n_{o}i}%
}\left(  \boldsymbol{\eta}\right)  \right]  $ is the Fisher information and
$p=\lim_{n\rightarrow \infty}\sqrt{\frac{n_{0}}{n}}$. \ A consistent estimator
of $I$ is given by the observed information:\
\[
\widehat{I}=-\frac{1}{n_{o}}\frac{\partial}{\partial \boldsymbol{\eta}^{\top}%
}Q_{n_{o}}\left(  \boldsymbol{\eta}\right)  \bigg|_{\boldsymbol{\eta=}%
\widehat{\boldsymbol{\eta}}}%
\]

By substituting $\widehat{\boldsymbol{\eta}}$ in place of $\boldsymbol{\eta}$,
we obtain:
\begin{equation}
U_{n}\left(  \boldsymbol{\beta},\widehat{\boldsymbol{\eta}}\right)
=\sum_{i=1}^{n}U_{ni}\left(  \boldsymbol{\beta},\widehat{\boldsymbol{\eta}%
}\right)  =\sum_{i=1}^{n}D_{i}V_{i}^{-1}S_{i}=\mathbf{0} \label{eqn270}%
\end{equation}
We solve (\ref{eqn270}) for $\boldsymbol{\beta}$ to obtain a consistent and
asymptotically normal estimator $\widehat{\boldsymbol{\beta}}$, or
$\widehat{\boldsymbol{\beta}}\left(  \widehat{\boldsymbol{\eta}}\right)  $.
\ However, unlike the case of known\ $\boldsymbol{\eta}$, the IF,
$\boldsymbol{\varphi}_{i}\left(  y_{i},\mathbf{w}_{i};\boldsymbol{\beta}%
_{0}\right)  =B^{-1}D_{i}V_{i}^{-1}S_{i}$, is no longer efficient and the
asymptotic variance $\Sigma_{\beta}$ for $\widehat{\boldsymbol{\beta}}\left(
\boldsymbol{\eta}_{0}\right)  $ in (\ref{eqn264}) is no longer the asymptotic
variance for $\widehat{\boldsymbol{\beta}}\left(  \widehat{\boldsymbol{\eta}%
}\right)  $. \ The asymptotic variance of the efficient version of
$\widehat{\boldsymbol{\beta}}\left(  \widehat{\boldsymbol{\eta}}\right)  $ is
given in Theorem \ref{thm1}. \ See Appendix for a derivation
of the efficient IF for $\widehat{\boldsymbol{\beta}}\left(  \widehat
{\boldsymbol{\eta}}\right)  $ and associated asymptotic variance. \ For
notational simplicity, we continue to denote $\widehat{\boldsymbol{\beta}%
}\left(  \widehat{\boldsymbol{\eta}}\right)  $ by $\widehat{\boldsymbol{\beta
}}\ $unless stated otherwise. \

\begin{theorem}
\label{thm1}Let $\widehat{\boldsymbol{\beta}}$ be the solution to
(\ref{eqn270}). \ Let
\begin{align}
D_{i}  &  =\frac{\partial}{\partial \boldsymbol{\beta}^{\top}}g_{i}\left(
\mathbf{w}_{i}\mathbf{;}\boldsymbol{\beta},\boldsymbol{\eta}\right)  ,\quad
S_{i}=y_{i}-g_{i},\quad V_{i}=Var\left(  y_{i}\mid \mathbf{w}_{i}\right)
,\label{eqn280}\\
M_{i}  &  =\frac{\partial}{\partial \boldsymbol{\eta}^{\top}}g_{i}\left(
\mathbf{w}_{i}\mathbf{;}\boldsymbol{\beta},\boldsymbol{\eta}\right)  ,\quad
H=E\left(  \frac{\partial}{\partial \boldsymbol{\eta}}Q_{i}\right)  ,\quad
p=\lim_{n\rightarrow \infty}\sqrt{\frac{n_{o}}{n}}.\nonumber
\end{align}
Then, under mild regularity conditions $\widehat{\boldsymbol{\beta}}$ is
consistent and asymptotically normal:
\[
\sqrt{n}\left(  \widehat{\boldsymbol{\beta}}-\boldsymbol{\beta}_{0}\right)
\rightarrow_{d}N(0,\Sigma_{\beta}),
\]
where
\begin{align}
\Sigma_{\beta}  &  =B^{-1}\left(  \Sigma_{U}+\Phi \right)  B^{-\top},\quad
B=E\left(  D_{i}^{\top}V_{i}^{-1}D_{i}\right)  ,\quad \Sigma_{U}=E\left(
D_{i}^{\top}V_{i}^{-1}S_{i}S_{i}^{\top}V_{i}^{-1}D_{i}\right)  ,
\label{eqn300}\\
\quad \Phi &  =\frac{1}{p^{4}}E\left(  D_{i}^{\top}V_{i}^{-1}M_{i}\right)
H^{-1}E\left(  M_{i}^{\top}V_{i}^{-\top}D_{i}\right)  ,\quad \nonumber
\end{align}

\end{theorem}

A consistent estimator of $\Sigma_{\beta}$ is obtained
by replacing the expectations in (\ref{eqn300}) with their respective
sample counterparts, akin to (\ref{eqn268}). \ 

Note that $B^{-1}\Sigma_{U}B^{-1}$ is the asymptotic variance for
$\widehat{\boldsymbol{\beta}}\left(  \boldsymbol{\eta}_{0}\right)  $, while
$B^{-1}\Phi B^{-\top}$ is the variance of the projection of the IF\ for
$\widehat{\boldsymbol{\beta}}\left(  \widehat{\boldsymbol{\eta}}\right)  $
onto the nuisance tangent space spanned by the score $Q{_{i}}\left(
\boldsymbol{\eta}\right)  $ (e.g., Chapter 4 of Tsiatis
\cite{tsiatis2006semiparametric}). \ We can also view $B^{-1}\Phi B^{-\top}$
as the additional variability in $\widehat{\boldsymbol{\beta}}\left(
\widehat{\boldsymbol{\eta}}\right)  $ induced by $\widehat{\boldsymbol{\eta}}%
$ due to their variational dependence. \ 

For semiparametric models with an infinite dimensional $\boldsymbol{\eta}$,
{\normalsize we consider the model in (\ref{eqn54}) for the transformed
right-censored }$t_{i}$ of $x_{i}$ {\normalsize under some }monotone
decreasing function $t=T^{-1}(x)$. \ Let $\nu=T^{-1}(\delta)$ denoting the transformed
detection limit $\delta$ for the left-censored $x_i$ to the right-censored $t_{i}$. \ As in Kong and Nan \cite{kong2016semiparametric},
we first estimate $\boldsymbol{\gamma}=\left(  \gamma_{0},\boldsymbol{\gamma
}_{1}^{\top}\right)  $ and then estimate $\xi \left(  \cdot \right)  $ given $\widehat{\boldsymbol{\gamma}}$ nonparametrically using the
Kaplan-Meier (KM) survival function $\widehat{S}\left(  t\right)  $ based on
the residuals $\widehat{\epsilon}_{i}\left(  t\right)  =t-\left(
\widehat{\gamma}_{0}+\widehat{\boldsymbol{\gamma}}_{1}^{\top}\mathbf{z}%
_{i}\right)  $. \ In this case, $\widehat{\xi}\left(  t\right)  $ represents
the empirical CDF\ with $\widehat{\xi}\left(  t\right)  =1-\widehat{S}\left(
t\right)  $. \ 

Given $\widehat{\boldsymbol{\eta}}=\left(  \widehat{\boldsymbol{\gamma}}%
^{\top},\widehat{\xi}\left(  \cdot \right)  \right)  ^{\top}$, we use the same
GEE in (\ref{eqn270}), but with slightly different $\widetilde{h}_{i}$, $D_{i}$ and $S_{i}$ to reflect the changed $\widetilde{h}_{i}$:
\begin{equation}
\widetilde{h}_{i}=\widetilde{h}\left(  \boldsymbol{\beta},\widehat
{\boldsymbol{\eta}}\right)  =\int_{\nu-(\widehat{\gamma}_{0}+\widehat
{\boldsymbol{\gamma}}_{1}^{\top}\mathbf{z}_{i})}^{\tau}g^{-1}((1,T(t+\widehat
{\gamma}_{0}+\widehat{\boldsymbol{\gamma}}_{1}^{\top}\mathbf{z}_{i}%
),\mathbf{u}_{i}^{\top})\boldsymbol{\beta})d\widehat{\xi}(t), \label{eqn340}%
\end{equation}
where $\tau$ is a constant truncation time to rule out the trivial case when the integral value is zero. \ 

As in the parametric case, the GEE estimator
$\widehat{\boldsymbol{\beta}}$ from (\ref{eqn340}) is consistent and
asymptotically normal, as summarized by the following Theorem \ref{thm2} (see
Appendix for a proof). \ 

\begin{theorem}
\label{thm2}Consider models \eqref{eqn54} and \eqref{eqn340}, and denote the
true value of $\boldsymbol{\beta}$ by $\boldsymbol{\beta}_{0}$. \ Under mild
regularity conditions, the GEE estimator $\widehat{\boldsymbol{\beta}}$ of
\eqref{eqn340} converges in outer probability to $\boldsymbol{\beta}_{0}$ and
$n^{1/2}(\widehat{\boldsymbol{\beta}}-\boldsymbol{\beta}_{0})$ converges
weakly to a zero-mean normal random variable. \ 
\end{theorem}

Unlike the semi-para combination, the asymptotic variance $\Sigma_{\beta}$ of
$\widehat{\boldsymbol{\beta}}$ is much more complicated for the
semi-semi method. \ For their proposed para-semi combination,
\cite{kong2016semiparametric} recommended the use of bootstrap to facilitate
estimation of the asymptotic variance. \ A more computationally efficient
alternative for the proposed approach is
to leverage sample-splitting and cross-fitting (SSCF) \cite{chernozhukov2018double}. \ For
example, we may split the observed sample into 2 subsamples with sample size $n_1$ and $n_2$ and use the second subsample to obtain $\widehat{\boldsymbol{\eta
}}_{2}$ and the first to estimate $\widehat{\boldsymbol{\beta}}_{1}$. \ Since $\widehat{\boldsymbol{\eta}}_2$ and $\widehat{\boldsymbol{\beta}}_1$ are independent, the asymptotic variance of $\widehat{\boldsymbol{\beta}}_1$ is the variance of the IF $\boldsymbol{\varphi}_i\left(y_i, \mathbf{w}_i ; \boldsymbol{\beta}_0, \boldsymbol{\eta}_0\right)$, which is consistently estimated by:
\[
\widehat{\Sigma}_{\boldsymbol{\beta}_{1}}^{(1)}=\frac{1}{n_{1}}\sum_{i=1}^{n_{1}}%
\boldsymbol{\varphi}_{i}\left(  y_{i},\mathbf{w}_{i};\widehat
{\boldsymbol{\beta}}_{1},\widehat{\boldsymbol{\eta}}_{2}\right)
\boldsymbol{\varphi}_{i}^{\top}\left(  y_{i},\mathbf{w}_{i};\widehat
{\boldsymbol{\beta}}_{1},\widehat{\boldsymbol{\eta}}_{2}\right)  =\frac
{1}{n_{1}}\sum_{i=1}^{n_{1}}\widehat{D}_{i}^{\top}\widehat{V}_{i}^{-1}%
\widehat{S}_{i}\widehat{S}_{i}^{\top}\widehat{V}_{i}^{-1}\widehat{D}_{i}%
\]
By switching the role of the two subsamples, we obtain $\widehat{\Sigma}_{\boldsymbol{\beta}_{1}}^{(2)}$ and
estimate the variance for the total sample by a weighted average of $\widehat{\Sigma}_{\boldsymbol{\beta}_{1}}^{(1)}$ and $\widehat{\Sigma}_{\boldsymbol{\beta}_{1}}^{(2)}$. \ 

Compared with the para-semi as in \cite{kong2016semiparametric}, where the score equations for estimating $\boldsymbol{\beta}$ involves integraing out missing $x_i$ in the complete-data score, the proposed semi-semi only performs this integration for the conditional mean $h\left(  x_{i},\mathbf{u}_{i};\boldsymbol{\beta}\right)$, making it much easier to implement. It is this feature coupled with SSCF that enables the approach to take a tiny fraction of computational time needed for the semi-semi alternative to provide inference about $\boldsymbol{\beta}$.   

With Theorem 1 and 2, we can test any linear hypotheses concerning components
of $\boldsymbol{\beta}$ by :
\begin{equation}
H_{0}:C\boldsymbol{\beta}=\mathbf{b}\quad \text{vs.}\quad H_{a}%
:C\boldsymbol{\beta}\neq \mathbf{b,} \label{eqn360}%
\end{equation}
where $C$ is a matrix and $\mathbf{b}$ is a vector of known constants. \ For
example, the null $H_{0}:\boldsymbol{\beta=0}$ for the primary model in
(\ref{eqn10}) can be expressed in this form by setting $\mathbf{C=I}_{p}$
where $p$ denotes the dimension of $\boldsymbol{\beta}$ and $\mathbf{I}_{p}$
the $p\times p$ identity matrix. \ The null of the linear hypothesis in
(\ref{eqn360}) is readily tested using, say the Wald statistic:\
\begin{equation}
W_{n}=n\left(  C\widehat{\boldsymbol{\theta}}-\mathbf{b}\right)  ^{\top
}\left(  C\widehat{\mathbf{\Sigma}}_{\boldsymbol{\theta}}C^{\top}\right)
^{-1}\left(  C\widehat{\boldsymbol{\theta}}-\mathbf{b}\right)  \rightarrow
_{d}\chi_{s}^{2}, \label{eqn380}%
\end{equation}
where $s$ is the rank of $C$ and $\chi_{s}^{2}$ denotes a (central) $\chi^{2}$
distribution with $s$ degrees of freedom. \ 

\section{Simulation Study\label{sec4}}

We start by assessing robustness of $\boldsymbol{\beta}$ by considering two
different parametric models for the error distribution for the primary
component under correctly specified parametric model for the second component
in (\ref{eqn70}). \ In all examples, we set type I error $\alpha=0.05$ and
Monte Carlo (MC)\ sample size $M=1,000$. \ 

\subsection{Robustness of Semiparametric Model for Primary
Component\label{sec4.1}}

For the primary component, we consider two exploratory variables with no
DL-effected missing, one continuous $u_{1i}$ and one binary $u_{2i}$, and with
errors following a normal and centered chi-square distribution:
\begin{align}
y_{i}  &  =\beta_{0}+\beta_{1}x_{i}+\beta_{2}u_{1i}+\beta_{3}u_{2i}%
+\epsilon_{yi},\text{\quad}u_{1i}\sim N\left(  \mu_{u},\sigma_{u}^{2}\right)
,\text{\quad}u_{2i}\sim Bern\left(  p_{u}\right)  ,\label{eqn400}\\
&  \epsilon_{yi}\overset{iid}{\sim}N\left(  0,\sigma_{y}^{2}\right)
\text{\quad or\quad}\epsilon_{yi}\overset{iid}{\sim}\left(  \chi_{s}%
^{2}-s\right)  \sqrt{\frac{\sigma_{y}^{2}}{2}},\nonumber
\end{align}
where $Bern\left(  p\right)  $ denotes a Bernoulli with mean $p$, $\chi
_{s}^{2}$ a chi-square with $s$ degrees of freedom. \ For the auxiliary
component, we consider a continuous $z_{1i}$ and a binary $z_{2i}$ explanatory
variable, and a normal error:
\begin{align}
x_{i}  &  =\gamma_{0}+\gamma_{1}z_{1i}+\gamma_{2}z_{2i}+\epsilon_{xi}%
,\quad \epsilon_{xi}\sim N\left(  0,\sigma_{x}^{2}\right)  ,\label{eqn420}\\
&  z_{1i}\sim N\left(  \mu_{z},\sigma_{z}^{2}\right)  ,\text{\quad}z_{2i}\sim
Bern\left(  p_{z}\right)  .\nonumber
\end{align}
We set sample size at $n=100$, $200$ and $500$ to represent small, moderate
and large samples, and missing percentage $\Pr \left(  x_{i}\leq \delta \right)
$ at 10\%, 30\% and 60\% to represent low, moderate and high levels of
missingness. \ All the other model parameters are set as:
\begin{align}
\beta_{0}  &  =\beta_{1}=\beta_{2}=\beta_{3}=1,\quad \gamma_{0}=\gamma
_{1}=\gamma_{2}=1,\quad \sigma_{y}^2=1,\quad \sigma_{x}^2=0.2,\label{eqn440}%
\\
\mu_{u}  &  =0,\quad \sigma_{u}^2=1,\quad p_{u}=0.5,\quad \mu_{z}%
=0,\quad \sigma_{z}^2=1,\quad p_{z}=0.5.\nonumber
\end{align}

Let $\widehat{\boldsymbol{\beta}}^{\left(  m\right)  }$ denote the estimated
$\boldsymbol{\beta}=\left(  \beta_{0},\beta_{1},\beta_{2},\beta_{3}\right)
^{\top}$ and $\widehat{\mathbf{\Sigma}}_{\beta}^{\left(  m\right)  }$ the
estimated asymptotic variance from the $m$th MC iteration. \ The MC estimator
$\widehat{\boldsymbol{\beta}}$ and associated asymptotic variance
$\widehat{\mathbf{\Sigma}}_{\beta}^{\left(  asymp\right)  }$ are the averaged
values of the respective $\widehat{\boldsymbol{\beta}}^{\left(  m\right)  }$
and $\widehat{\mathbf{\Sigma}}_{\beta}^{\left(  m\right)  }$ over the $M$ MC
iterations. \ The empirical variance $\widehat{\mathbf{\Sigma}}_{\beta
}^{\left(  emp\right)  }$ is estimated as the sample variance of
$\widehat{\boldsymbol{\beta}}^{\left(  m\right)  }$ across the MC iterations. \ 

We applied the approach for the semi-para combination for inference about the parameters. \ For space
considerations, we only report results for inference about $\beta_{1}$. \ In
addition to comparing the estimates and variances, we also assess type I
errors from testing the null:\ $H_{0}:\beta_{1}=\beta_{10}$, where $\beta
_{10}$ denotes the true value of $\beta_{1}$ in (\ref{eqn440}). \ The Wald
statistic $w_{n}^{\left(  m\right)  }=\frac{\widehat{\beta}_{1}^{\left(
m\right)  }-\beta_{10}}{\sqrt{\widehat{\mathbf{\sigma}}_{\beta_{1}}^{2\left(
m\right)  }}}$ at the $m$th MC iteration is approximately a standard normal,
where $\widehat{\mathbf{\sigma}}_{\beta_{1}}^{2\left(  m\right)  }$ denotes
the asymptotic variance of $\widehat{\beta}_{1}^{\left(  m\right)  }$ under
$H_{0}$. \ Thus the type I\ error rate over MC\ iterations, $\widehat{\alpha
}^{W}=\frac{1}{M}\sum_{m=1}^{M}I\left(  w_{n}^{\left(  m\right)  }\geq
q_{1,0.95}\right)  $, should be close to the nominal level $\alpha=0.05$ at
least for large samples, where $q_{s,0.95}$ denotes the $95$th percentile of a
central $\chi_{s}^{2}$. \ 

Shown in Table \ref{tab1} are MC\ estimates of $\beta_{1}$, along with
asymptotic and empirical standard errors, and type I\ errors for testing the
$H_{0}:\beta_{1}=\beta_{10}$. \ For both error distributions and all three
sample sizes, the estimates $\widehat{\boldsymbol{\beta}}_{1}$ were quite
close to the true $\beta_{10}=1$. \ The asymptotic and empirical standard
errors were also close to each other for both error distributions even with
60\% missing data, especially for larger $n$. \ The type I\ error rates were
close to the nominal value as well. \  \begin{table}[ptb]
\caption{\ MC estimates of $\beta_{1}$, standard errors (asymptotic and
empirical), and type I\ errors under the null. \ }%
\label{tab1}
\begin{center}
$%
\begin{tabular}
[c]{|c|c|c|c|}\hline
\multicolumn{4}{|c|}{Estimates of $\beta_{1}$, Standard Errors and Type I
Errors for Testing $H_{0}:$ $\beta_{1}=1$}\\ \hline
\multicolumn{4}{|c|}{With 10\%/30\%/60\% Missing Data due to Detection
Limit}\\ \hline
Estimate & \multicolumn{2}{|c}{Standard Error $\left(  \times10^{-2}\right)
$} & Testing $H_{0}$\\ \hline
& Asymptotic & Empirical & Type I error\\ \hline
\multicolumn{4}{|c|}{$n=100$}\\ \hline
\multicolumn{4}{|c|}{Normal Error}\\ \hline
$0.993/1.002/1.011$ & $0.730/1.020/3.513$ & $0.729/1.074/3.00$ &
$0.059/0.067/0.047$\\ \hline
\multicolumn{4}{|c|}{Centered Chi-square Error}\\ \hline
$0.997/1.002/1.006$ & $0.743/1.019/3.646$ & $0.776/0.989/2.686$ &
$0.057/0.052/0.038$\\ \hline
\multicolumn{4}{|c|}{$n=200$}\\ \hline
\multicolumn{4}{|c|}{Normal Error}\\ \hline
$0.997/0.997/1.005$ & $0.368/0.498/1.566$ & $0.397/0.499/1.467$ &
$0.060/0.060/0.055$\\ \hline
\multicolumn{4}{|c|}{Centered Chi-square Error}\\ \hline
$1.000/1.000/0.994$ & $0.363/0.501/1.552$ & $0.394/0.518/1.303$ &
$0.050/0.048/0.047$\\ \hline
\multicolumn{4}{|c|}{$n=500$}\\ \hline
\multicolumn{4}{|c|}{Normal Error}\\ \hline
$0.998/0.998/1.000$ & $0.148/0.198/0.571$ & $0.150/0.206/0.546$ &
$0.047/0.063/0.051$\\ \hline
\multicolumn{4}{|c|}{Centered Chi-square Error}\\ \hline
$0.998/0.997/1.002$ & $0.148/0.197/0.566$ & $0.148/0.197/0.566$ &
$0.064/0.052/0.042$\\ \hline
\end{tabular}
\  \  \  \ $
\end{center}
\end{table}

\subsection{Comparison of Parametric and Semiparametric Model for Auxiliary
Component\label{sec4.2}}

We continue to assume (\ref{eqn400}) for the primary component.
\ For the auxiliary component, we also considered a semiparametric model for the
transformed $t_{i}=-\log \left(  x_{i}\right)  $:%
\begin{align}
t_{i}  &  =\gamma_{0}+\gamma_{1}z_{1i}+\gamma_{2}z_{2i}+\epsilon_{xi}%
,\quad \epsilon_{xi}\overset{iid}{\sim}N\left(  0,\sigma_{x}^{2}\right)
,\label{eqn_t}\\
&  z_{1i}\overset{iid}{\sim}Bern\left(  0.5\right)  ,\text{\quad}%
z_{2i}\overset{iid}{\sim}N\left(  1,1\right)  ,\nonumber
\end{align}
with monotone decreasing transformation $T(t_{i})=\exp(-t_{i})$. We set sample size at $n=200$ and $400$ and missing percentage at 30\%. \ All
the other parameters are set as:
\[
\beta_{1}=2\quad \beta_{2}=0.5,\quad \beta_{0}=\beta_{3}=-1,\quad \gamma
_{0}=\gamma_{1}=0.25,\quad \gamma_{2}=-0.5,\quad \sigma_{y}^2=1,\quad
\sigma_{x}^2=0.01
\]

Shown in Table \ref{tab_compare} are the bias, asymptotic and empirical
standard errors of the parameters for the two semi-semi and semi-para methods
from MC\ simulations. \ For comparison purposes, Table \ref{tab_compare} also
shows these values from applying the primary model in (\ref{eqn10}) to the full
and observed data. \ For both sample sizes, the biases from all methods were
close to 0, although the ones from the observed data approach were slightly larger.
\ From the table, the empirical standard errors from the semi-para method were
slightly smaller than their semi-semi counterparts. \ For the semi-para model,
the empirical standard errors were also quite close to their asymptotic
counterparts. \ 

\  \begin{table}[ptb]
\caption{\ Comparison of bias, asymptotic and empirical standard errors of the
proposed semi-semi and semi-para methods with full and observed data based on
MC estimates of $\beta_{1}$. }%
\label{tab_compare}
\begin{center}
$%
\begin{tabular}
[c]{|c|c|c|c|}\hline
\multicolumn{4}{|c|}{Comparison of Semi-semi and Semi-para}\\ \hline
\multicolumn{4}{|c|}{Models with 30\% Missing Data due to Detection
Limit}\\ \hline
& Bias & \multicolumn{2}{|c|}{Standard Error}\\ \hline
&  & Asymptotic & Empirical\\ \hline
\multicolumn{4}{|c|}{$n=200$}\\ \hline
Full Data & $0.0038$ & $0.276$ & $0.296$\\ \hline
Semi-semi & $0.0017$ & $0.283$ & $0.299$\\ \hline
Semi-para & $0.0053$ & $0.277$ & $0.297$\\ \hline
Observed Data & $0.0068$ & $0.470$ & $0.484$\\ \hline
\multicolumn{4}{|c|}{$n=400$}\\ \hline
Full Data & $-0.0080$ & $0.205$ & $0.202$\\ \hline
Semi-semi & $-0.0101$ & $0.209$ & $0.205$\\ \hline
Semi-para & $-0.0071$ & $0.207$ & $0.204$\\ \hline
Observed Data & $-0.0135$ & $0.329$ & $0.329$\\ \hline
\end{tabular}
\  \  \  \  \  \ $
\end{center}
\end{table}

\subsection{Efficiency Gains\label{sec4.3}}

We now investigate efficiency gains under different parameter settings through
power analysis. \ To evaluate power, we continue to simulate data from
(\ref{eqn400}) and (\ref{eqn420}). \ Since efficiency gains depend on
$\sigma_{x}^{2}$ relative to $\sigma_{y}^{2}$, we fix $\sigma_{y}^{2}=1$ and
vary $\sigma_{x}^{2}$ from $\sigma_{x}^{2}=0.1$ to $0.5$, $1$ and $5$, while
keeping all the other parameter values the same as in (\ref{eqn440}), except
for $\beta_{1}$, since we consider testing the hypothesis concerning this
parameter:\
\begin{equation}
H_{0}:\beta_{1}=0\text{\quad vs.\quad}H_{a}:\beta_{1}=b \label{eqn480}%
\end{equation}
where $b$ $\left(  \neq0\right)  $ is a known constant. \ By simulating data
from $\beta_{1}=1.1$ and testing the null $\beta_{1}=0$, we estimate power
by:
\[
\text{power estimate}=\frac{\text{Number of times }H_{0}\text{ is rejected }%
}{M}%
\]
where $M=1,000$ is the number of MC simulations. \ Since the semi-para
provides the most efficiency gains in real data analysis, we only report power
estimates for this method. \ 

Shown in Table \ref{tab2} are MC\ power estimates for testing the hypothesis
in (\ref{eqn480}) for the semi-para method and observed data analysis with
sample size $n=500$. \ Power estimates were quite similar between the normal
and centered chi-square errors. \ As expected, power decreased as the percent of
missing due to DL\ increased from 10\% to 30\% to 60\%. \ When the variance
ratio $\frac{\sigma_{x}^{2}}{\sigma_{y}^{2}}=0.1$, the proposed method
improved efficiency for all missingness levels. \ As $\frac{\sigma_{x}^{2}%
}{\sigma_{y}^{2}}$ increased, efficiency gains were still observed for 10\%
missing but diminished and even became negative for 30\% and 60\% missing.
\ This is expected as efficiency gains relies critically on the variance ratio
$\frac{\sigma_{x}^{2}}{\sigma_{y}^{2}}$. \ As $\frac{\sigma_{x}^{2}}%
{\sigma_{y}^{2}}$ increases, the variability of the estimator $\widehat{\beta
}_{1}$ increases, reducing the efficiency gains until the proposed method
eventually becomes less efficient. \  \begin{table}[ptb]
\caption{\ MC estimates of power for testing the null $H_{0}:\beta_{1}=0$
under the semi-para method. \ }%
\label{tab2}
\begin{center}
$%
\begin{tabular}
[c]{|c|c|c|}\hline
\multicolumn{3}{|c|}{Comparison of Power for Testing $H_{0}:\beta_{1}=0$ vs.
$\ H_{0}:\beta_{1}=b$}\\ \hline
\multicolumn{3}{|c|}{Between Semi-para Model and Observed Data Analysis}%
\\ \hline
\multicolumn{3}{|c|}{With 10\%/30\%/60\% Missing due to DL}\\ \hline
Error Distribution & Observed Data & Semi-para\\ \hline
& \multicolumn{2}{|c|}{$\frac{\sigma_{x}^{2}}{\sigma_{y}^{2}}=0.1$}\\ \hline
Normal & $0.556/0.332/0.161$ & $0.710/0.603/0.318$\\ \hline
Centered Chi-square & $0.572/0.336/0.140$ & $0.711/0.651/0.339$\\ \hline
& \multicolumn{2}{|c|}{$\frac{\sigma_{x}^{2}}{\sigma_{y}^{2}}=0.5$}\\ \hline
Normal Error & $0.668/0.433/0.178$ & $0.781/0.542/0.131$\\ \hline
Centered Chi-square & $0.671/0.433/0.182$ & $0.784/0.515/0.141$\\ \hline
& \multicolumn{2}{|c|}{$\frac{\sigma_{x}^{2}}{\sigma_{y}^{2}}=1$}\\ \hline
Normal Error & $0.767/0.517/0.223$ & $0.857/0.484/0.118$\\ \hline
Centered Chi-square & $0.789/0.527/0.209$ & $0.866/0.481/0.119$\\ \hline
& \multicolumn{2}{|c|}{$\frac{\sigma_{x}^{2}}{\sigma_{y}^{2}}=5$}\\ \hline
Normal Error & $0.988/0.896/0.524$ & $0.988/0.545/0.137$\\ \hline
Centered Chi-square & $0.993/0.925/0.541$ & $0.992/0.551/0.159$\\ \hline
\end{tabular}
\  \  \  \  \  \  \ $
\end{center}
\end{table}

\section{Real Study\label{sec5}}

The proposed approach was conducted within the Study of Secondary Exposures to Pesticides Among Children and Adolescents (ESPINA; Estudio de la Exposici\'{o}n Secundaria a Plaguicidas en
Ni\~{n}os y Adolescentes [Spanish]) \cite{suarez2012lower}. \ This ESPINA study is a prospective cohort study established in
2008 in the agricultural county of Pedro Moncayo, Ecuador, that aims to assess effects of pesticide exposures in children living in agricultural settings. \ One of the primary objectives of this analysis is
to examine association of acetylcholinestaerase (AChE) activity with
urinary concentrations of a phenoxy acid herbicide, 2,4-Dichlorophenoxyacetic acid (2,4-D), as our case study metabolite, for which about one-third of the samples are censored due to DL. \ For our illustration, the data included 523 participants with information on urinary concentrations of 2,4-D, other non-missing metabolites, and relevant covariates.

For the\ primary association between AChE and 2,4-D, the SPGLM in
(\ref{eqn10}) includes log transformed 2,4-D (log(2,4-D)) as
well as other explanatory variables including age, gender, race, BMI-for-age z-score
(ZBA), height-for-age z-score (ZHA) and hemoglobin (Hgb). \ Unlike 2,4-D, none
of the other explanatory variables is subject to missing due to DL. \ For the
auxiliary component, we model $x=$ log(2,4-D)   using both the
parametric (\ref{eqn50}) with normal errors and semiparametric (\ref{eqn54})
with the monotone decreasing transformation $x=T(t)=-t$, with the explanatory
variables: creatinine, age, race, ZHA, log-distance (log transformed distance
between the subject's residence and the nearest agricultural crop) and cohabitation (a
binary indicating if any of the subject's family members worked in
plantations). \ 

Shown in Table \ref{tab3} are estimated coefficients, along with
their standard errors and p-values, for the parametric (\ref{eqn50}) and
semiparametric (\ref{eqn54}) auxiliary components. \ The estimates and standard
errors were slightly different between the two models, but with the same
directions and statistical
significance for all the variables. \ Only creatinine and age had significant
associations with (log(2,4-D)) in both models. \ Even though the
remaining variables were not significant, we chose to include them based on a
priori content knowledge. \  \begin{table}[ptb]
\caption{\ Estimates of parameters, standard errors and p-values from
parametric and semiparametric auxiliary models used to model log(2,4-D) with creatinine, age, race, ZHA, log-distance and cohabitation in ESPINA study. \ }%
\label{tab3}
\begin{center}%
\begin{tabular}
[c]{|l|rrr|rrr|}\hline
Explanatory variables & \multicolumn{3}{c|}{Parametric model} &
\multicolumn{3}{c|}{Semiparametric model}\\ \hline
& Estimate & Std. Error & p-value & Estimate & Std. Error & p-value\\ \hline
Creatinine & $-0.0057$ & $0.0007$ & $<0.001$ & $-0.0079$ & $0.0012$ &
$<0.001$\\
Age & $0.0868$ & $0.0255$ & $<0.001$ & $0.0923$ & $0.0363$ & $0.011$\\
Race & $-0.0313$ & $0.1076$ & $0.771$ & $-0.1036$ & $0.1493$ & $0.488$\\
ZHA & $0.0706$ & $0.0474$ & $0.137$ & $0.0821$ & $0.0811$ & $0.312$\\
Log-Distance & $-0.0422$ & $0.0273$ & $0.122$ & $-0.0763$ & $0.0419$ &
$0.069$\\
Cohabitation & $-0.0191$ & $0.0918$ & $0.835$ & $-0.0478$ & $0.1303$ &
$0.713$\\ \hline
\end{tabular}
$\  \  \  \  \  \  \ $
\end{center}
\end{table}

Shown in Table \ref{tab4} are estimated coefficients, along with
standard errors and p-values, for the primary regression model obtained from the observed
data, semi-para, and semi-semi methods. \ All estimates from the two methods
were in the same directions and similar to their counterparts from the
observed data analysis, except for ZHA and ZBA, which were not significant in
all methods. \ For 2,4-D, the standard errors of the estimates and associated
p-values from both proposed methods were smaller than their counterparts from the
observed data analysis, indicating efficiency gains from the proposed approach.
\ The most notable difference was the conclusion about the effect of 2,4-D on
AChE; while it was not significant in the observed data\ analysis ($p=0.122$),
it became quite significant in the semi-para ($p=0.019$) and borderline
significant in the semi-semi ($p=0.076$) method. \ We also applied Kong and Nan's approach \cite{kong2016semiparametric} to our data analysis and obtained similar results as semi-semi method. \ However, their method required approximately 7.5 hours to complete, while ours took less than one minute.

\begin{table}[ptb]
\caption{\ Estimates of parameters, standard errors and p-values from
primary models used to model AChE and log(2,4-D), age, gender, race, ZHA, ZBA and Hgb in ESPINA study based on the two methods and observed data analysis.\ }%
\label{tab4}
\begin{center}
$%
\begin{tabular}
[c]{|c|c|c|c|c|c|c|c|}\hline
\multicolumn{8}{|c|}{Estimates of Parameters, Standard Errors, and p-values
for}\\ 
\multicolumn{8}{|c|}{Primary Model}\\ \hline
Explanatory & age & gender & race & ZHA & ZBA & Hgb & log(2,4-D)\\ \hline
variables &  &  &  &  &  &  & \\ \hline
& \multicolumn{7}{|c|}{Observed Data}\\ \hline
Estimate & $0.042$ & $0.124$ & $0.116$ & $0.003$ & $-0.018$ & $0.271$ &
$0.046$\\ \hline
Std. Error & $0.013$ & $0.046$ & $0.053$ & $0.024$ & $0.027$ & $0.020$ &
$0.030$\\ \hline
p-value & $0.001$ & $0.007$ & $0.028$ & $0.910$ & $0.498$ & $<0.001$ &
$0.122$\\ \hline
& \multicolumn{7}{|c|}{Semi-para}\\ \hline
Estimate & $0.034$ & $0.115$ & $0.130$ & $-0.023$ & $0.020$ & $0.270$ &
$0.052$\\ \hline
Std. Error & $0.012$ & $0.041$ & $0.043$ & $0.025$ & $0.022$ & $0.022$ &
$0.022$\\ \hline
p-value & $0.004$ & $0.005$ & $0.002$ & $0.357$ & $0.367$ & $<0.001$ &
$0.019$\\ \hline
& \multicolumn{7}{|c|}{Semi-semi}\\ \hline
Estimate & $0.036$ & $0.112$ & $0.130$ & $-0.018$ & $0.017$ & $0.270$ &
$0.037$\\ \hline
Std. Error & $0.011$ & $0.041$ & $0.043$ & $0.025$ & $0.022$ & $0.023$ &
$0.021$\\ \hline
p-value & $0.002$ & $0.006$ & $0.003$ & $0.467$ & $0.454$ & $<0.001$ &
$0.076$\\ \hline
\end{tabular}
\  \  \  \  \  \  \  \ $
\end{center}
\end{table}

\section{Discussion\label{sec6}}

In this paper, we developed a new approach for semiparametric regression analysis with a
DL-effected explanatory variable. \ By extending the primary component of
interest to semiparametric GLM, the proposed approach provides more robust
inference over existing alternatives. \ This extension is more critical than extending the auxiliary model to semiparametric regression as in Kong and Nan \cite{kong2016semiparametric}, since estimators based on the observed data are consistent for the primary component, if conditional mean model is correctly specified. \ We also considered two variations of the auxiliary component, allowing for both semiparametric and parametric models for left-censored data. \ Although the semi-semi combination is
more robust, the semi-para model is also a useful alternative. \ As the
primary objective of using surrogates is to improve power, the semi-para method
provides more efficiency gains, as illustrated by both the simulated and real
data examples. \ In practice, one may try different regression models for the
auxiliary component. \ If\ estimates of regression coefficients for the
primary component are similar between the semi-semi and one or more of the
semi-para models, one may report results from the semi-para model with most
similar parameter estimates for efficiency gains. \ Work is underway to
develop goodness of fit statistics to formalize such a model selection
process. \ 

In this work we leveraged SSCF to address the computational challenge for inference. \ Compared to bootstrap, SSCF provides inference about parameters of interest for the primary component within our setting, rather than all parameters for both components. \ This advantage arises because it is difficult to analytically derive the efficient IF for the parameters of interest in the primary component with their variational dependence on the nuisance parameters for the semiparametric auxiliary component. \ SSCF circumvents such difficulty by estimating variational-dependent parameters using two independent subsample. \ This approach significantly improves computational efficiency over bootstrap by 450 times. \ Indeed, without using SSCF, it would have been really difficult to examine performance of the proposed approach using Monte Carlo simulations, as in Kong and Nan \cite{kong2016semiparametric} in which only empirical variances were reported.  \

We focused on one explanatory variable with DL-effected missing and
cross-sectional data. \ In many real studies, there are often multiple such
explanatory variables in regression analysis. \ By using an auxiliary model
for each of such variables, we may extend the approach to this general
setting. \ By modeling the marginal mean of a longitudinal response at each
assessment using the semiparametric model in (\ref{eqn10}), we may also extend
the approach to longitudinal data setting, with inference based on GEE or
weighted GEE for longitudinal data \cite{tang2023applied}. \ In comparison, it is
much more difficult to extend a parametric primary model to longitudinal data
using parametric generalized linear mixed-effects models (GLMM) because of
computational challenges due to high-dimensional integration. \ For example,
popular software packages such as R\ and SAS\ do not even provide reliable
inference when GLMM\ is applied to longitudinal binary responses \cite{chen2016power,tang2023applied,zhang2011fitting,lin2023modelling}. \ Work is underway to develop these extensions. \

\pagebreak

\newpage
\appendix

\section*{ Web Appendix A. Derivation of the expectation of truncated normal}

\label{appendix.a} The truncated density of $x_{i}$ for the model
(6)\ with a normal error is given by:
\[
f\left(  x_{i}\mid \mathbf{z}_{i},x_{i}>\delta \right)  =\frac{\phi_{\mu
_{i},\sigma_{x}^{2}}\left(  x\right)  I\left(  x>\delta \right)  }{1-\Phi
_{\mu_{i},\sigma_{x}^{2}}\left(  \delta \right)  }%
\]
Thus we have:\
\begin{align*}
E\left(  x_{i}\mid \mathbf{z}_{i},x_{i}>\delta \right)   &  =\int_{-\infty
}^{\infty}xf\left(  x\mid \mathbf{z}_{i},x_{i}\leq \delta \right)  dx=\int
_{\delta}^{\infty}\frac{x\phi_{\mu_{i},\sigma_{x}^{2}}\left(  x\right)
}{1-\Phi_{\mu_{i},\sigma_{x}^{2}}\left(  \delta \right)  }dx\\
&  =\frac{1}{1-\Phi_{\mu_{i},\sigma_{x}^{2}}\left(  \delta \right)  }\left(
\int_{\delta}^{\infty}(x-\mu_{i})\phi_{\mu_{i},\sigma_{x}^{2}}\left(
x\right)  dx+\int_{\delta}^{\infty}\mu_{i}\phi_{\mu_{i},\sigma_{x}^{2}}\left(
x\right)  dx\right)  .
\end{align*}
Since $\frac{d}{dx}\phi_{\mu_{i},\sigma_{x}^{2}}\left(  x\right)
=-\frac{(x-\mu_{i})}{\sigma_{x}^{2}}\phi_{\mu_{i},\sigma_{x}^{2}}\left(
x\right)  $, the above simplifies to:
\begin{align*}
E\left(  x_{i}\mid \mathbf{z}_{i},x_{i}>\delta \right)  &=\frac{1}{1-\Phi
_{\mu_{i},\sigma_{x}^{2}}\left(  \delta \right)  }\left(  \int_{\delta}%
^{\infty}-\sigma_{x}^{2}\frac{d}{dx}\phi_{\mu_{i},\sigma_{x}^{2}}\left(
x\right)  dx+\mu_{i}\left(  1-\Phi_{\mu_{i},\sigma_{x}^{2}}\left(
\delta \right)  \right)  \right) \\
& =\mu_{i}+\sigma^{2}\frac{\phi_{\mu_{i}%
,\sigma_{x}^{2}}(\delta)}{1-\Phi_{\mu_{i},\sigma_{x}^{2}}\left(
\delta \right)  }.
\end{align*}

Similarly, we have:%
\[
E\left(  x_{i}\mid \mathbf{z}_{i},x_{i}\leq \delta \right)  =\frac{1}{\Phi
_{\mu_{i},\sigma_{x}^{2}}\left(  \delta \right)  }\left(  \int_{-\infty
}^{\delta}-\sigma_{x}^{2}\frac{d}{dx}\phi_{\mu_{i},\sigma_{x}^{2}}\left(
x\right)  dx+\mu_{i}\Phi_{\mu_{i},\sigma_{x}^{2}}\left(  \delta \right)
\right)  =\mu_{i}-\sigma_{x}^{2}\frac{\phi_{\mu_{i},\sigma_{x}^{2}}(\delta
)}{\Phi_{\mu_{i},\sigma_{x}^{2}}(\delta)}.
\]

\section*{ Web Appendix B. Proof of Theorem 1.}

\label{appendix.b} Let $\widehat{\mathbf{\eta}}$ denote the MLE of
$\mathbf{\eta}$ from solving the score equations:
\[
Q_{n_{o}}\left(  \mathbf{\eta}\right)  =\sum_{x_{i}>\delta}Q_{i}%
=\mathbf{0}%
\]
\ By the asymptotic linear property of the MLE, we have:%
\[
\sqrt{n_{o}}\left(  \widehat{\mathbf{\eta}}-\mathbf{\eta}\right)
=-H^{-1}\frac{\sqrt{n_{o}}}{n_{o}}Q_{n_{o}}\left(  \mathbf{\eta}\right)
+\mathbf{o}_{p}\left(  1\right)  =-H^{-1}\frac{\sqrt{n_{o}}}{n_{o}}\sum
_{x_{i}>\delta}Q_{i}+\mathbf{o}_{p}\left(  1\right)  .
\]
where $H=E\left(  \frac{\partial}{\partial \mathbf{\eta}}Q_{i}\right)  $
and $\mathbf{o}_{p}\left(  \cdot \right)  $\ denotes the stochastic
$\mathbf{o}\left(  \cdot \right)  $ (Kowalski and Tu, 2007). \ By applying a
Taylor series expansion to $U_{n}\left(  \widehat{\mathbf{\beta}}%
,\widehat{\mathbf{\eta}}\right)  $ in (11) in the main text,\ we have:\
\begin{equation}
\frac{\sqrt{n}}{n}U_{n}=\left(  -\frac{1}{n}\frac{\partial}{\partial
\mathbf{\beta}^{\top}}U_{n}\right)  ^{\top}\sqrt{n}\left(  \widehat
{\mathbf{\beta}}-\mathbf{\beta}\right)  +\left(  -\frac{1}{n}%
\frac{\partial}{\partial \mathbf{\eta}^{\top}}U_{n}\right)  ^{\top}\sqrt
{n}\left(  \widehat{\mathbf{\eta}}-\mathbf{\eta}\right)
+\mathbf{o}_{p}\left(  1\right)  \label{a.1}%
\end{equation}
Since
\begin{align*}
-\frac{1}{n}\frac{\partial}{\partial \mathbf{\beta}^{\top}}U_{n} &
\rightarrow_{p}E\left(  D_{i}^{\top}V_{i}^{-1}D_{i}\right)  =B,\\
-\frac{1}{n}\frac{\partial}{\partial \mathbf{\eta}^{\top}}U_{n} &
\rightarrow_{p}-E\left[  \frac{\partial}{\partial \mathbf{\eta}^{\top}%
}\left(  D_{i}^{\top}V_{i}^{-1}S_{i}\right)  \right]  =E\left[  D_{i}^{\top
}V_{i}^{-1}\frac{\partial}{\partial \mathbf{\eta}^{\top}}g_{i}\left(
\mathbf{w}_{i}\mathbf{;}\mathbf{\beta},\mathbf{\eta}\right)  \right]
\\
&  =E\left(  D_{i}^{\top}V_{i}^{-1}M_{i}\right)  =C
\end{align*}
(\ref{a.1})\ can be expressed as:%
\[
\frac{\sqrt{n}}{n}U_{n}=B^{\top}\sqrt{n}\left(  \widehat{\mathbf{\beta}%
}-\mathbf{\beta}\right)  +C^{\top}\sqrt{n}\left(  \widehat
{\mathbf{\eta}}-\mathbf{\eta}\right)  +\mathbf{o}_{p}\left(  1\right)
\]
Solving the above for $\sqrt{n}\left(  \widehat{\mathbf{\beta}%
}-\mathbf{\beta}\right)  $ yields:\
\begin{align}
\sqrt{n}\left(  \widehat{\mathbf{\beta}}-\mathbf{\beta}\right)   &
=B^{-1}\left(  \frac{\sqrt{n}}{n}U_{n}-C^{\top}\sqrt{n}\left(  \widehat
{\mathbf{\eta}}-\mathbf{\eta}\right)  \right)  +\mathbf{o}_{p}\left(
1\right)  \label{a.2}\\
&  =B^{-1}\left(  \frac{\sqrt{n}}{n}\sum_{i=1}^{n}U_{ni}-C^{\top}\frac
{\sqrt{n}}{\sqrt{n_{0}}}\sqrt{n_{o}}\left(  \widehat{\mathbf{\eta}%
}-\mathbf{\eta}\right)  \right)  +\mathbf{o}_{p}\left(  1\right)
\nonumber \\
&  =B^{-1}\left(  \frac{\sqrt{n}}{n}\sum_{i=1}^{n}U_{ni}+C^{\top}\frac
{\sqrt{n}}{\sqrt{n_{0}}}\frac{\sqrt{n_{o}}}{n_{o}}H^{-1}\sum_{x_{i}>\delta
}Q_{i}\right)  +\mathbf{o}_{p}\left(  1\right)  \nonumber \\
&  =B^{-1}\left(  \frac{\sqrt{n}}{n}\sum_{i=1}^{n}U_{ni}+C^{\top}\frac
{\sqrt{n}}{n}\frac{n}{n_{o}}H^{-1}\sum_{i=1}^{n}Q_{i}I\left(  x_{i}%
>\delta \right)  \right)  +\mathbf{o}_{p}\left(  1\right)  \nonumber \\
&  =\frac{\sqrt{n}}{n}\sum_{i=1}^{n}B^{-1}\left(  U_{ni}+C^{\top}p^{-2}%
H^{-1}Q_{i}I\left(  x_{i}>\delta \right)  \right)  +\mathbf{o}_{p}\left(
1\right)  \nonumber \\
&  =\frac{\sqrt{n}}{n}\sum_{i=1}^{n}\mathbf{\varphi}_{i}\left(
y_{i},\mathbf{w}_{i};\mathbf{\beta}\right)  +\mathbf{o}_{p}\left(
1\right)
\end{align}
where $p=\lim_{n\rightarrow \infty}\frac{\sqrt{n_{o}}}{\sqrt{n}}$. \ Thus
$\mathbf{\varphi}_{i}\left(  y_{i},\mathbf{w}_{i};\mathbf{\beta
}\right)  =$ $B^{-1}\left(  U_{ni}+p^{-2}C^{\top}H^{-1}Q_{i}I\left(
x_{i}>\delta \right)  \right)  $ is the influence function and the asymptotic
normality follows from the central limit theorem and Slutsky's theorem, with
the asymptotic variance $\Sigma_{\beta}=Var\left(  \mathbf{\varphi}%
_{i}\left(  y_{i},\mathbf{w}_{i};\mathbf{\beta}\right)  \right)  $. \ 

Note that $B^{-1}p^{-2}C^{\top}H^{-1}Q_{i}I\left(  x_{i}>\delta \right)  $ is
also the projection of $B^{-1}U_{ni}$ onto the nuisance tangent space of
$\mathbf{\eta}$ and thus $\mathbf{\varphi}_{i}\left(  y_{i}%
,\mathbf{w}_{i};\mathbf{\beta}\right)  $ is the efficient influence
function for $\widehat{\mathbf{\beta}}$ (Chapter 4 of \cite{tsiatis2006semiparametric}).

\section*{ Web Appendix C. Proof of Theorem 2.}
A semiparametric linear regression is used to model $t = T^{-1}(x)$, where $T$ is a strictly decreasing transformation such that its range is nonnegative:
\begin{equation} \label{eq:semiparametric}
t=\gamma_0+\mathbf{\gamma}_1^{\top} \mathbf{z}+\epsilon,
\end{equation}
where $\epsilon$ follows a non-parametric distribution characterized by an infinite-dimensional function $\xi(t)$. Let $\nu=T^{-1}(\delta)$ denote the right-censoring point for $t$. We first estimate $\mathbf{\gamma} = (\gamma_{0}, \mathbf{\gamma}_{1}^{\top})^{\top}$. Given the estimator $\widehat{\mathbf{\gamma}}=(\widehat{\gamma}_{0},\widehat{\mathbf{\gamma}}_{1}^{\top})^{\top}$, we estimate $\xi$ using the nonparametric Kaplan-Meier (KM) survival curve applied to the residuals $\epsilon(t)=t-(\widehat{\gamma}_{0}+\widehat{\mathbf{\gamma}}_{1}^{\top}\mathbf{z})$. In this case, $\widehat{\xi}(t)$ represents the empirical cumulative distribution function (CDF) of $\xi(t) = 1 - S(t)$, where the survival function $S(t)$ is estimated by the KM estimator. Given the estimators $\widehat{\mathbf{\gamma}}$ and $\widehat{\xi}$, we estimate $\mathbf{\beta}$ using the generalized estimating equations (GEE):
\begin{equation} \label{eq:GEE}
\mathbf{U}_n(\mathbf{\beta}, \widehat{\mathbf{\gamma}}, \widehat{\xi})=\frac{1}{n}\sum_{i=1}^{n}\left\{\mathbbm{1}(t_{i} \leq \nu)  U_{i}(\mathbf{\beta})+\mathbbm{1}(t_{i} > \nu)  \widetilde{U}_{i}(\mathbf{\beta}, \widehat{\mathbf{\gamma}}, \widehat{\xi})\right\}=\mathbf{0},
\end{equation}
where
\begin{align*}
U_{i}(\mathbf{\beta}) &= D_i(\mathbf{\beta}) V_i^{-1} S_i(\mathbf{\beta}), \quad \widetilde{U}_{i}(\mathbf{\beta}, \mathbf{\gamma}, \xi) =\widetilde{D}_{i}(\mathbf{\beta}, \mathbf{\gamma}, \xi) \widetilde{V}_i^{-1} \widetilde{S}_{i}(\mathbf{\beta}, \mathbf{\gamma}, \xi), \\
D_i(\mathbf{\beta}) &= \frac{\partial}{\partial \mathbf{\beta}^{\top}} h_i(\mathbf{\beta}), \quad S_i(\mathbf{\beta})=y_i-h_i(\mathbf{\beta}), \quad V_i=\operatorname{Var}(y_i \mid t_i, \mathbf{u}_i), \\
\widetilde{D}_{i}(\mathbf{\beta}, \mathbf{\gamma}, \xi) &= \frac{\partial}{\partial \mathbf{\beta}^{\top}}\widetilde{h}_i(\mathbf{\beta}, \mathbf{\gamma}, \xi), \quad \widetilde{S}_{i}(\mathbf{\beta}, \mathbf{\gamma}, \xi)=y_{i}-\widetilde{h}_i(\mathbf{\beta}, \mathbf{\gamma}, \xi), \quad \widetilde{V}_i=\operatorname{Var}(y_i \mid t_i = E(t \mid t \leq \nu), \mathbf{u}_i), \\
h_i(\mathbf{\beta}) &= g^{-1}((1, T(t_i), \mathbf{u}_i^{\top}) \mathbf{\beta}), \quad \widetilde{h}_{i}(\mathbf{\beta}, \mathbf{\gamma}, \xi) = \int_{\nu-(\gamma_{0}+\mathbf{\gamma}_{1}^{\top}\mathbf{z}_{i})}^{\tau}g^{-1}((1, T(t+\gamma_{0}+\mathbf{\gamma}_{1}^{\top}\mathbf{z}_{i}), \mathbf{u}_i^{\top}) \mathbf{\beta}) d\xi(t),
\end{align*}
and $g$ is the link function.

Note that $\tau$ is a constant truncation time defined in Condition 4 in Regularity Conditions. From Condition 4 in Regularity Conditions, $t$ is truncated by some $\tau < \nu$. Therefore, the existence of $\tau$ rules out the trivial case where the integral value is zero when no association exists between $t$ and $\mathbf{u}$. Given $\widehat{\mathbf{\gamma}}$ and the KM survival function estimator $\widehat{\xi}(t)$,
\[
\widetilde{h}_{i}(\mathbf{\beta}, \widehat{\mathbf{\gamma}}, \widehat{\xi}) = \sum_{i=1}^n g^{-1}((1, T(t+\gamma_{0}+\mathbf{\gamma}_{1}^{\top}\mathbf{z}_{i}), \mathbf{u}_i^{\top}) \mathbf{\beta})
\mathbbm{1}(\nu-(\widehat{\gamma}_{0}+\widehat{\mathbf{\gamma}}_{1}^{\top}\mathbf{z}_{i}) \leq t_i \leq \tau).
\]

As in the parametric case, the GEE estimator $\widehat{\mathbf{\beta}}$ is
consistent and asymptotically normal.

\subsection*{ Web Appendix C.1. Regularity Conditions. }\label{sec:C.1}

Define a deterministic function
\[
\mathbf{U}(\mathbf{\beta}, \mathbf{\gamma}, \xi)=E[\mathbbm{1}(t \leq \nu) U(\mathbf{\beta})+\mathbbm{1}(t > \nu) \widetilde{U}(\mathbf{\beta}, \mathbf{\gamma}, \xi)],
\]
and denote the true value of $(\mathbf{\beta}, \mathbf{\gamma}, \xi)$ by $(\mathbf{\beta}_0, \mathbf{\gamma}_0, \xi_0)$. Let $\mathcal{Y} \subset \mathbb{R}$ denote the sample space of the response variable $y$, $\mathcal{T} \subset \mathbb{R}$ the sample space of $t = T^{-1}(x)$, $\mathcal{U} \subset \mathbb{R}^{p-2}$ the sample space of the covariate $\mathbf{u}$, $\Theta \subset \mathbb{R}^p$ the parameter space of $\mathbf{\beta}$, $\Gamma \subset \mathbb{R}^q$ the parameter space of $\mathbf{\gamma}$, and $\Xi$ the parameter space of $\xi$. In addition to the assumptions of bounded support for $(t, \mathbf{u}, y)$ and compactness of the parameter spaces $\Theta$ and $\Gamma$, we state a set of regularity conditions for Theorem \ref{thm:1} in Section 2:

\textit{Condition} 1. $\mathbf{U} (\mathbf{\beta}, \mathbf{\gamma}_0, \xi_0)$ has a unique root $\mathbf{\beta}_0$;

\textit{Condition} 2. for any constant $K<\infty, \sup_{t \in[\nu, K]}|T(t)| \leqslant c_0<\infty, \sup_{t \in[\nu, K]}|T'(t)| \leqslant c_1<\infty$, and $\sup_{t \in[\nu, K]}|T''(t)| \leqslant c_2<\infty$, where $T'$ and $T''$ are the first and second derivatives of $T$, respectively, and $c_0, c_1$ and $c_2$ are constants;

\textit{Condition} 3. $\epsilon$ has bounded density $\xi'$ with bounded derivative $\xi''$; in other words, $\xi' \leq c_3 < \infty$ and $|\xi''| \leq c_4 < \infty$ for constants $c_3$ and $c_4$, and
\[
\int_{-\infty}^{\infty}\{\xi''(t) / \xi'(t)\}^2 \xi'(t) d t<\infty;
\]

\textit{Condition} 4. there is a constant truncation time $\tau < \infty$ such that
\[
P(\min(t, \nu)-(\gamma_0+\mathbf{\gamma}_1^{\top}\mathbf{z}) > \tau \mid \mathbf{u}) \geq \zeta > 0
\]
for all $\mathbf{u} \in \mathcal{U}$ and $\mathbf{\gamma} \in \Gamma$;

\textit{Condition} 5. $g$ is a bounded monotone function such that $g^{-1}$ has a bounded first derivative $(g^{-1})'$ and a bounded Lipschitz second derivative $(g^{-1})''$;

\iffalse
\textit{Condition} 6. there exist constants $C_i(i=1, 2, 3)$ such that for any constant $K<\infty$,
\begin{align*}
& \sup_{\mathbf{\beta} \in \Theta, \mathbf{u} \in \mathcal{U}, t \in [\nu, K]} |h(t, \mathbf{u}; \mathbf{\beta})| \leq C_1<\infty, \\
& \sup_{\mathbf{\beta} \in \Theta, \mathbf{u} \in \mathcal{U}, t \in [\nu, K]} \left|\frac{\partial h(t, \mathbf{u}; \mathbf{\beta})}{\partial t}\right| \leq C_2<\infty, \\
& \sup_{\mathbf{\beta} \in \Theta, \mathbf{u} \in \mathcal{U}, t \in [\nu, K]} \left|\left[\frac{\partial h(t, \mathbf{u}; \mathbf{\beta})}{\partial \mathbf{\beta}}\right]_j\right| \leq C_3<\infty, \quad 1 \leq j \leq p;
\end{align*}
\fi

\textit{Condition} 6. there exist constants $\varepsilon>0$ such that
$|\mathbf{\gamma}-\mathbf{\gamma}_0| + \|\xi-\xi_0\| < \varepsilon$, where $|\cdot|$ is the Euclidean norm and $\|\cdot\|$ is the $\ell^{\infty}$ norm;

\textit{Condition} 7. there exists a neighborhood $N \subset \Gamma \times \Xi$ of $(\mathbf{\gamma}_0, \xi_0)$ such that
\begin{itemize}
\item[(1)] for each $(\mathbf{\gamma}, \xi) \in N$, the map $\mathbf{\beta} \mapsto \mathbf{U}(\mathbf{\beta}, \mathbf{\gamma}, \xi)$ is differentiable, and its derivative $D_{\mathbf{\beta}}\mathbf{U}(\mathbf{\beta}, \mathbf{\gamma}, \xi) \in \mathbb{R}^p$ is jointly continuous in $(\mathbf{\gamma}, \xi)$ over $N$;

\item[(2)] for each fixed $\mathbf{\beta} \in \Theta$ and $(\mathbf{\gamma}, \xi) \in N$, the map $\mathbf{\gamma} \mapsto \widetilde{U}(\mathbf{\beta}, \mathbf{\gamma}, \xi)$ is differentiable, and its derivative $D_{\mathbf{\gamma}}\widetilde{U}(\mathbf{\beta}, \mathbf{\gamma}, \xi) \in \mathbb{R}^q$ is continuous in $\xi$ over $N$;

\item[(3)] for each fixed $\mathbf{\beta} \in \Theta$ and $(\mathbf{\gamma}, \xi) \in N$, the map $\xi \mapsto \widetilde{U}(\mathbf{\beta}, \mathbf{\gamma}, \xi)$ is Fréchet differentiable; i.e., for every $(\mathbf{\gamma}, \xi) \in N$ satisfying $|\mathbf{\gamma} - \mathbf{\gamma}_0| + \|\xi - \xi_0\| < \varepsilon$, there exists a bounded linear operator $D_{\xi}\widetilde{U}(\mathbf{\beta}, \mathbf{\gamma}, \xi) \in \Xi^*$ such that
\[
\frac{|\widetilde{U}(\mathbf{\beta}, \mathbf{\gamma}, \xi + h_\xi) - \widetilde{U}(\mathbf{\beta}, \mathbf{\gamma}, \xi) - D_{\xi}\widetilde{U}(\mathbf{\beta}, \mathbf{\gamma}, \xi)(h_\xi)|}{\|h_\xi\|} \to 0
\]
as $\|h_\xi\| \to 0$, where $\Xi^*$ is the dual space of $\Xi$;
\end{itemize}

\textit{Condition} 8. for each fixed $\mathbf{\beta} \in \Theta$ and $(\mathbf{\gamma}, \xi) \in N$, there exists an integrable function $M(t, \mathbf{u}, y)$ such that for all small $h_\xi$,
\[
\left| \frac{\mathbf{U}(\mathbf{\beta}, \mathbf{\gamma}, \xi + h_\xi)(t, \mathbf{u}, y) - \mathbf{U}(\mathbf{\beta}, \mathbf{\gamma}, \xi)(t, \mathbf{u}, y) - D_{\xi} \mathbf{U}(\mathbf{\beta}, \mathbf{\gamma}, \xi)(h_\xi)(t, \mathbf{u}, y)}{\|h_\xi\|} \right| \leq M(t, \mathbf{u}, y),
\]
and $\int M(t, \mathbf{u}, y) \, dP(t, \mathbf{u}, y) < \infty$.

\subsection*{ Web Appendix C.2. Asymptotic Properties }\label{sec:C.2}

\begin{theorem} \label{thm:1}
Consider models \eqref{eq:semiparametric} and \eqref{eq:GEE}, and denote the true value of $\mathbf{\beta}$ by $\mathbf{\beta}_0$. Suppose that the regularity conditions in Appendix C.1. hold. Then the GEE estimator $\widehat{\mathbf{\beta}}$ satisfying $\mathbf{U}_n(\widehat{\mathbf{\beta}}, \widehat{\mathbf{\gamma}}, \widehat{\xi})=0$ converges in outer probability to $\mathbf{\beta}_0$, and $n^{1 / 2}(\widehat{\mathbf{\beta}}-\mathbf{\beta}_0)$ converges weakly to a zero-mean normal random variable.
\end{theorem}

The proof of Theorem \ref{thm:1} is based on the general Z-estimation theory of \cite{kong2016semiparametric} and \cite{nan2013general}, which is provided in the following Lemmas \ref{lm:1} and \ref{lm:2} for our problem setting. Define $\rho\{(\mathbf{\gamma}, \xi),(\mathbf{\gamma}_0, \xi_0)\}=\left|\mathbf{\gamma}-\mathbf{\gamma}_0\right|+\left\|\xi-\xi_0\right\|$. We use $P^*$ to denote outer probability, which is defined as $P^*(A)=\inf \{P(B): B \supset A, B \in \mathcal{B}\}$ for any subset $A$ of $\Omega$ in a probability space $(\Omega, \mathcal{B}, P)$.

\subsection*{ Web Appendix C.2.1 Consistency }

\setcounter{theorem}{0}
\begin{lemma} \label{lm:1}
(Consistency) Suppose $\mathbf{\beta}_0$ is the unique solution to $\mathbf{U} (\mathbf{\beta}, \mathbf{\gamma}_0, \xi_0)=0$ in the parameter space $\Theta$ and $(\widehat{\mathbf{\gamma}}, \widehat{\xi})$ are estimators of $(\mathbf{\gamma}_0, \xi_0)$ such that $\rho\{(\widehat{\mathbf{\gamma}}, \widehat{\xi}),(\mathbf{\gamma}_0, \xi_0)\}=o_{p^*}(1)$. If
\[
\sup_{\mathbf{\beta} \in \Theta, \rho\{(\mathbf{\gamma}, \xi), (\mathbf{\gamma}_0, \xi_0)\} \leq \varepsilon_n} \frac{\left|\mathbf{U}_n(\mathbf{\beta}, \mathbf{\gamma}, \xi)-\mathbf{U}(\mathbf{\beta}, \mathbf{\gamma}_0, \xi_0)\right|}{1+\left|\mathbf{U}_n(\mathbf{\beta}, \mathbf{\gamma}, \xi)\right|+\left|\mathbf{U}(\mathbf{\beta}, \mathbf{\gamma}_0, \xi_0)\right|}=o_{p^*}(1)
\]
for every sequence $\left\{\varepsilon_n \downarrow 0\right\}$, then $\widehat{\mathbf{\beta}}$ satisfying $\mathbf{U}_n(\widehat{\mathbf{\beta}}, \widehat{\mathbf{\gamma}}, \widehat{\xi})=o_{p^*}(1)$ converges in outer probability to $\mathbf{\beta}_0$.
\end{lemma}

\begin{proof}[Proof of consistency in Theorem \ref{thm:1}]
We prove consistency using Lemma \ref{lm:1}. Since $\widehat{\mathbf{\gamma}}$ is $n^{1/2}$-consistent \cite{Nan2009,ding2011sieve}  and $\widehat{\xi}$ is also $n^{1/2}$-consistent over a finite interval, as shown in Lemma 3 of \cite{kong2016semiparametric}, we have
\[
\rho\{(\widehat{\mathbf{\gamma}}, \widehat{\xi}),(\mathbf{\gamma}_0, \xi_0)\}=o_{p^*}(1).
\]
Given that $\mathbf{\beta}_0$ is the unique solution to $\mathbf{U} (\mathbf{\beta}, \mathbf{\gamma}_0, \xi_0)=0$ from Condition 1 of regularity conditions, we only need to show that
\[
\sup_{\mathbf{\beta} \in \Theta, \rho\{(\mathbf{\gamma}, \xi),(\mathbf{\gamma}_0, \xi_0)\} \leq \varepsilon_n}|\mathbf{U}_n (\mathbf{\beta}, \mathbf{\gamma}, \xi)-\mathbf{U} (\mathbf{\beta}, \mathbf{\gamma}_0, \xi_0)|=o_{p^*}(1)
\]
for every sequence $\varepsilon_n \downarrow 0$. Note that
\begin{align}
&\sup_{\mathbf{\beta} \in \Theta, \rho\{(\mathbf{\gamma}, \xi),(\mathbf{\gamma}_0, \xi_0)\} \leq \varepsilon_n}|\mathbf{U}_n (\mathbf{\beta}, \mathbf{\gamma}, \xi)-\mathbf{U} (\mathbf{\beta}, \mathbf{\gamma}_0, \xi_0)| \nonumber\\
&\quad \leq \sup _{\mathbf{\beta} \in \Theta}|(\mathbb{P}_n-P)(\mathbbm{1}(t \leq \nu)U(\mathbf{\beta}))| \label{eq:3}\\
&\quad \quad + \sup _{\mathbf{\beta} \in \Theta, \rho\{(\mathbf{\gamma}, \xi),(\mathbf{\gamma}_0, \xi_0)\} \leq \varepsilon_n} |(\mathbb{P}_n-P)(\mathbbm{1}(t > \nu)\widetilde{U}(\mathbf{\beta}, \mathbf{\gamma}, \xi))| \label{eq:4}\\
&\quad \quad + \sup _{\mathbf{\beta} \in \Theta, \rho\{(\mathbf{\gamma}, \xi),(\mathbf{\gamma}_0, \xi_0)\} \leq \varepsilon_n} P|\widetilde{U}(\mathbf{\beta}, \mathbf{\gamma}, \xi)-\widetilde{U}(\mathbf{\beta}, \mathbf{\gamma}_0, \xi_0)|. \label{eq:5}
\end{align}
It suffices to show that \eqref{eq:3}, \eqref{eq:4}, and \eqref{eq:5} are equal to $o_{p^*}(1)$.

We first prove that \eqref{eq:5} is equal to $o_{p^*}(1)$. By the mean value theorem and Conditions 7 (2)-(3), we obtain
\begin{align*}
\widetilde{U}(\mathbf{\beta}, \mathbf{\gamma}, \xi)-\widetilde{U}(\mathbf{\beta}, \mathbf{\gamma}_0, \xi_0)&=\widetilde{U}(\mathbf{\beta}, \mathbf{\gamma}, \xi)-\widetilde{U}(\mathbf{\beta}, \mathbf{\gamma}_0, \xi)+\widetilde{U}(\mathbf{\beta}, \mathbf{\gamma}_0, \xi)-\widetilde{U}(\mathbf{\beta}, \mathbf{\gamma}_0, \xi_0) \\
&= D_{\mathbf{\gamma}}\widetilde{U}(\mathbf{\beta}, \mathbf{\gamma}', \xi) (\mathbf{\gamma}-\mathbf{\gamma}_0)+D_{\xi}\widetilde{U}(\mathbf{\beta}, \mathbf{\gamma}_0, \xi') (\xi-\xi_0),
\end{align*}
where $\mathbf{\gamma}' = (1-c_1)\mathbf{\gamma} + c_1\mathbf{\gamma}_0$ and $\xi' = (1-c_2)\xi + c_2\xi_0$ for some $c_1, c_2 \in (0, 1)$. Thus,
\begin{align*}
&\sup_{\mathbf{\beta} \in \Theta, \rho\{(\mathbf{\gamma}, \xi),(\mathbf{\gamma}_0, \xi_0)\} \leq \varepsilon_n} P|\widetilde{U}(\mathbf{\beta}, \mathbf{\gamma}, \xi)-\widetilde{U}(\mathbf{\beta}, \mathbf{\gamma}_0, \xi_0)|\\
&\quad \leq \sup _{\mathbf{\beta} \in \Theta, \rho\{(\mathbf{\gamma}, \xi),(\mathbf{\gamma}_0, \xi_0)\} \leq \varepsilon_n} \left(P|D_{\mathbf{\gamma}}\widetilde{U}(\mathbf{\beta}, \mathbf{\gamma}', \xi)|+P\|D_{\xi}\widetilde{U}(\mathbf{\beta}, \mathbf{\gamma}_0, \xi')\|\right)\varepsilon_n \\
&\quad = o_{p^*}(1).
\end{align*}

Next, we demonstrate that \eqref{eq:3} is equal to $o_{p^*}(1)$. Consider the class of function
\begin{equation} \label{eq:6}
\{\mathbbm{1}(t \leq \nu)U(\mathbf{\beta}): t \in \mathcal{T}, \mathbf{\beta} \in \Theta\}.
\end{equation}
By Exercise 14 on page 152 in \cite{van1996weak} , $\{(t, \mathbf{u}^{\top})^{\top} \mapsto \beta_0+\beta_1T(t)+\mathbf{\beta}_2^{\top}\mathbf{u}: t \leq \nu, \mathbf{u} \in \mathbb{R}^{p-2}, (\beta_0, \beta_1, \mathbf{\beta}_2^{\top})^{\top} \in \mathbb{R}^p\}$ is a VC-class of dimension $p$, so it is Donsker. From Condition 5, $(g^{-1})'$ and $(g^{-1})''$ are bounded, so $g^{-1}$ and $(g^{-1})'$ are Lipschitiz. Therefore, by Theorem 2.10.6 in \cite{van1996weak}, $\{h(\mathbf{\beta}): \mathbf{\beta} \in \Theta\}$ and $\{D(\mathbf{\beta}): \mathbf{\beta} \in \Theta\}$ are Donsker, and thus the function class \eqref{eq:6} is Donsker. Hence, $\eqref{eq:3}$ is equal to $o_{p^*}(1)$.

Finally, we verify that \eqref{eq:4} is equal to $o_{p^*}(1)$. Consider the class of function
\begin{equation} \label{eq:7}
\{\widetilde{U}(\mathbf{\beta}, \mathbf{\gamma}, \xi): \mathbf{\beta} \in \Theta, \mathbf{\gamma} \in \Gamma, \xi \in \Xi, \rho\{(\mathbf{\gamma}, \xi),(\mathbf{\gamma}_0, \xi_0)\} \leq \varepsilon_n\}.
\end{equation}
Note that $\widetilde{U}(\mathbf{\beta}, \mathbf{\gamma}, \xi) = \widetilde{D}(\mathbf{\beta}, \mathbf{\gamma}, \xi)\widetilde{V}^{-1}\widetilde{S}(\mathbf{\beta}, \mathbf{\gamma}, \xi)$, where
\begin{align*}
\widetilde{D}(\mathbf{\beta}, \mathbf{\gamma}, \xi) &= \int_{\nu-(\gamma_{0}+\mathbf{\gamma}_{1}^{\top}\mathbf{z})}^{\tau}\frac{\partial}{\partial \mathbf{\beta}^{\top}}g^{-1}((1, T(t+\gamma_{0}+\mathbf{\gamma}_{1}^{\top}\mathbf{z}), \mathbf{u}^{\top}) \mathbf{\beta}) d\xi(t) \\
&= \int_{\nu-(\gamma_{0}+\mathbf{\gamma}_{1}^{\top}\mathbf{z})}^{\tau}(g^{-1})'((1, T(t+\gamma_{0}+\mathbf{\gamma}_{1}^{\top}\mathbf{z}), \mathbf{u}^{\top}) \mathbf{\beta})(1, T(t+\gamma_{0}+\mathbf{\gamma}_{1}^{\top}\mathbf{z}), \mathbf{u}^{\top}) d\xi(t) \\
&= (g^{-1})'((1, T(\tau+\gamma_{0}+\mathbf{\gamma}_{1}^{\top}\mathbf{z}), \mathbf{u}^{\top})\mathbf{\beta})(1, T(\tau+\gamma_{0}+\mathbf{\gamma}_{1}^{\top}\mathbf{z}), \mathbf{u}^{\top}) \xi(\tau) \\
&\quad -(g^{-1})'((1, T(\nu), \mathbf{u}^{\top})\mathbf{\beta})(1, T(\nu), \mathbf{u}^{\top}) \xi(\nu-(\gamma_{0}+\mathbf{\gamma}_{1}^{\top}\mathbf{z})) \\
&\quad -\int_{\nu-c_5}^{\tau} \mathbbm{1}(t \geq \nu - (\gamma_{0}+\mathbf{\gamma}_{1}^{\top}\mathbf{z}))\xi(t) \\
&\quad \quad \quad \left[(g^{-1})''((1, T(t+\gamma_{0}+\mathbf{\gamma}_{1}^{\top}\mathbf{z}), \mathbf{u}^{\top})\mathbf{\beta})\beta_1T'(t+\gamma_{0}+\mathbf{\gamma}_{1}^{\top}\mathbf{z}) (1, T(t+\gamma_{0}+\mathbf{\gamma}_{1}^{\top}\mathbf{z}), \mathbf{u}^{\top}) \right.\\
&\quad \quad \quad \quad \left. + \; (g^{-1})'((1, T(t+\gamma_{0}+\mathbf{\gamma}_{1}^{\top}\mathbf{z}), \mathbf{u}^{\top})\mathbf{\beta})(0, T'(t+\gamma_{0}+\mathbf{\gamma}_{1}^{\top}\mathbf{z}), \mathbf{0}^{\top}) \right] dt,
\end{align*}
\begin{align*}
\widetilde{S}(\mathbf{\beta}, \mathbf{\gamma}, \xi) &= y - \int_{\nu-(\gamma_{0}+\mathbf{\gamma}_{1}^{\top}\mathbf{z})}^{\tau}g^{-1}((1, T(t+\gamma_{0}+\mathbf{\gamma}_{1}^{\top}\mathbf{z}), \mathbf{u}^{\top}) \mathbf{\beta}) d\xi(t) \\
&= y - g^{-1}((1, T(\tau+\gamma_{0}+\mathbf{\gamma}_{1}^{\top}\mathbf{z}), \mathbf{u}^{\top})\mathbf{\beta}) \xi(\tau) + g^{-1}((1, T(\nu), \mathbf{u}^{\top})\mathbf{\beta}) \xi(\nu-(\gamma_{0}+\mathbf{\gamma}_{1}^{\top}\mathbf{z})) \\
&\quad + \; \int_{\nu-c_5}^{\tau} \mathbbm{1}(t \geq \nu-(\gamma_{0}+\mathbf{\gamma}_{1}^{\top}\mathbf{z})) \xi(t) \\
&\quad \quad \quad (g^{-1})'((1, T(t+\gamma_{0}+\mathbf{\gamma}_{1}^{\top}\mathbf{z}), \mathbf{u}^{\top})\mathbf{\beta})\beta_1T'(t+\gamma_{0}+\mathbf{\gamma}_{1}^{\top}\mathbf{z}) dt,
\end{align*}
and $c_5 = \sup_{\mathbf{\gamma} \in \Gamma, \mathbf{u} \in \mathcal{U}}|\gamma_{0}+\mathbf{\gamma}_{1}^{\top}\mathbf{z}| < \infty$. $T$ and $T'$ are Lipschitz by Condition 2 and $g^{-1}$, $(g^{-1})'$, and $(g^{-1})''$ are Lipschitz by Conditions 5, so $\{T(t): t \in \mathcal{T}\}$, $\{T'(t): t \in \mathcal{T}\}$, $\{g^{-1}((1, T(t+\gamma_{0}+\mathbf{\gamma}_{1}^{\top}\mathbf{z}), \mathbf{u}^{\top})\mathbf{\beta}): t \in \mathcal{T}, \mathbf{\beta} \in \Theta,  \mathbf{\gamma} \in \Gamma\}$, $\{(g^{-1})'((1, T(\tau+\gamma_{0}+\mathbf{\gamma}_{1}^{\top}\mathbf{z}), \mathbf{u}^{\top})\mathbf{\beta}): t \in \mathcal{T}, \mathbf{\beta} \in \Theta,  \mathbf{\gamma} \in \Gamma\}$, and $\{(g^{-1})''((1, T(\tau+\gamma_{0}+\mathbf{\gamma}_{1}^{\top}\mathbf{z}), \mathbf{u}^{\top})\mathbf{\beta}): t \in \mathcal{T}, \mathbf{\beta} \in \Theta,  \mathbf{\gamma} \in \Gamma\}$ belong to Donsker classes by Theorem 2.10.6 in \cite{van1996weak}. Moreover, by Exercise 9 on page 151 and Exercise 14 on page 152 in \cite{van1996weak}, we know that the class of indicator functions of half spaces is a VC-class, so $\{\mathbbm{1}(t \geq \nu-(\gamma_{0}+\mathbf{\gamma}_{1}^{\top}\mathbf{z})): t \in \mathcal{T}, (\gamma_0, \mathbf{\gamma}_1^{\top})^{\top} \in \Gamma\}$ is Donsker. By Lemma 5 in \cite{kong2016semiparametric}, we know that $\{\xi(\nu-(\gamma_0 + \mathbf{\gamma}_1^\top \mathbf{z})): (\gamma_0, \mathbf{\gamma}_1^{\top})^{\top} \in \Gamma\}$ is Donsker. By Theorem 2.10.3 in \cite{van1996weak} , the permanence of the Donsker property for the closure of the convex hull, we have
\begin{align}
\Bigg\{ &\int_{\nu-c_5}^{\tau} \mathbbm{1}(t \geq \nu - (\gamma_{0}+\mathbf{\gamma}_{1}^{\top}\mathbf{z}))\xi(t) \\
&\quad \quad \quad \left[(g^{-1})''((1, T(t+\gamma_{0}+\mathbf{\gamma}_{1}^{\top}\mathbf{z}), \mathbf{u}^{\top})\mathbf{\beta})\beta_1T'(t+\gamma_{0}+\mathbf{\gamma}_{1}^{\top}\mathbf{z}) (1, T(t+\gamma_{0}+\mathbf{\gamma}_{1}^{\top}\mathbf{z}), \mathbf{u}^{\top}) \right.\\
&\quad \quad \quad \quad \left. + \; (g^{-1})'((1, T(t+\gamma_{0}+\mathbf{\gamma}_{1}^{\top}\mathbf{z}), \mathbf{u}^{\top})\mathbf{\beta})(0, T'(t+\gamma_{0}+\mathbf{\gamma}_{1}^{\top}\mathbf{z}), \mathbf{0}^{\top}) \right] dt: \\
&\mathbf{\beta} \in \Theta, \mathbf{\gamma} \in \Gamma, \xi \in \Xi, \rho\{(\mathbf{\gamma}, \xi),(\mathbf{\gamma}_0, \xi_0)\} \leq \varepsilon_n\Bigg\}
\end{align}
and
\begin{align}
\Bigg\{ &\int_{\nu-c_5}^{\tau} \mathbbm{1}(t \geq \nu-(\gamma_{0}+\mathbf{\gamma}_{1}^{\top}\mathbf{z})) \xi(t) \cdot (g^{-1})'((1, T(t+\gamma_{0}+\mathbf{\gamma}_{1}^{\top}\mathbf{z}), \mathbf{u}^{\top})\mathbf{\beta})\beta_1T'(t+\gamma_{0}+\mathbf{\gamma}_{1}^{\top}\mathbf{z}) dt: \\
&\mathbf{\beta} \in \Theta, \mathbf{\gamma} \in \Gamma, \xi \in \Xi, \rho\{(\mathbf{\gamma}, \xi),(\mathbf{\gamma}_0, \xi_0)\} \leq \varepsilon_n\Bigg\}
\end{align}
are Donsker. Hence, $\widetilde{D}(\mathbf{\beta}, \mathbf{\gamma}, \xi)$ and $\widetilde{S}(\mathbf{\beta}, \mathbf{\gamma}, \xi)$ belong to Donsker classes, and it follows that \eqref{eq:4} is equal to $o_{p^*}(1)$.
\end{proof}

\subsection*{ Web Appendix C.2.2 Asymptotic Normality }

\begin{lemma} \label{lm:2}
(Rate of convergence and asymptotic representation) Suppose that $\widehat{\mathbf{\beta}}$ satisfying $\mathbf{U}_n(\widehat{\mathbf{\beta}}, \widehat{\mathbf{\gamma}}, \widehat{\xi})$ $= o_{p^*}(n^{-1 / 2})$ is a consistent estimator of $\mathbf{\beta}_0$ that is a solution to $\mathbf{U} (\mathbf{\beta}, \mathbf{\gamma}_0, \xi_0)=0$ in $\Theta$, and that $(\widehat{\mathbf{\gamma}}, \widehat{\xi})$ is an estimator of $(\mathbf{\gamma}_0, \xi_0)$ satisfying $\rho\{(\widehat{\mathbf{\gamma}}, \widehat{\xi}),(\mathbf{\gamma}_0, \xi_0)\}=O_{p^*}(n^{-1 / 2})$. Suppose the following four conditions are satisfied:

\begin{itemize}
\item[(i)] (stochastic equicontinuity)
\[
\frac{|n^{1 / 2}(\mathbf{U}_n-\mathbf{U})(\widehat{\mathbf{\beta}}, \widehat{\mathbf{\gamma}}, \widehat{\xi})-n^{1 / 2}(\mathbf{U}_n-\mathbf{U})(\mathbf{\beta}_0, \mathbf{\gamma}_0, \xi_0)|}{1+n^{1 / 2}|\mathbf{U}_n(\widehat{\mathbf{\beta}}, \widehat{\mathbf{\gamma}}, \widehat{\xi})|+n^{1 / 2}|\mathbf{U}(\widehat{\mathbf{\beta}}, \widehat{\mathbf{\gamma}}, \widehat{\xi})|}=o_{p^*}(1);
\]

\item[(ii)] $n^{1 / 2} \mathbf{U}_n(\mathbf{\beta}_0, \mathbf{\gamma}_0, \xi_0)=O_{p^*}(1)$;

\item[(iii)] (smoothness) there exist continuous matrices $D_{\mathbf{\beta}}\mathbf{U}(\mathbf{\beta}_0, \mathbf{\gamma}_0, \xi_0)$, $D_{\mathbf{\gamma}}\mathbf{U}(\mathbf{\beta}_0, \mathbf{\gamma}_0, \xi_0)$, and a continuous linear functional $D_{\xi}\mathbf{U}(\mathbf{\beta}_0, \mathbf{\gamma}_0, \xi_0)$ such that
\begin{align*}
&|\mathbf{U}(\widehat{\mathbf{\beta}}, \widehat{\mathbf{\gamma}}, \widehat{\xi})-\mathbf{U}(\mathbf{\beta}_0, \mathbf{\gamma}_0, \xi_0)-D_{\mathbf{\beta}}\mathbf{U}(\mathbf{\beta}_0, \mathbf{\gamma}_0, \xi_0)(\widehat{\mathbf{\beta}}-\mathbf{\beta}_0) \\
&\quad -D_{\mathbf{\gamma}}\mathbf{U}(\mathbf{\beta}_0, \mathbf{\gamma}_0, \xi_0)(\widehat{\mathbf{\gamma}}-\mathbf{\gamma}_0)-D_{\xi}\mathbf{U}(\mathbf{\beta}_0, \mathbf{\gamma}_0, \xi_0)(\widehat{\xi}-\xi_0)| \\
&=o(|\widehat{\mathbf{\beta}}-\mathbf{\beta}_0|)+o(\rho\{(\widehat{\mathbf{\gamma}}, \widehat{\xi}),(\mathbf{\gamma}_0, \xi_0)\}),
\end{align*}
here we assume that the matrix $D_{\mathbf{\beta}}\mathbf{U}(\mathbf{\beta}_0, \mathbf{\gamma}_0, \xi_0)$ is nonsingular; and

\item[(iv)] $n^{1 / 2}D_{\mathbf{\gamma}}\mathbf{U}(\mathbf{\beta}_0, \mathbf{\gamma}_0, \xi_0)(\widehat{\mathbf{\gamma}}-\mathbf{\gamma}_0)=O_{p^*}(1)$ and $n^{1 / 2}D_{\xi}\mathbf{U}(\mathbf{\beta}_0, \mathbf{\gamma}_0, \xi_0)(\widehat{\xi}-\xi_0)=O_{p^*}(1)$.
\end{itemize}

\noindent
Then $\widehat{\mathbf{\beta}}$ is $n^{1 / 2}$-consistent and
\begin{align*}
n^{1 / 2}(\widehat{\mathbf{\beta}}-\mathbf{\beta}_0) &= -\{D_{\mathbf{\beta}}\mathbf{U}\left(\mathbf{\beta}_0, \mathbf{\gamma}_0, \xi_0\right)\}^{-1} n^{1 / 2}\{\left(\mathbf{U}_n-\mathbf{U}\right)\left(\mathbf{\beta}_0, \mathbf{\gamma}_0, \xi_0\right) \\
&\quad \quad +D_{\mathbf{\gamma}}\mathbf{U}(\mathbf{\beta}_0, \mathbf{\gamma}_0, \xi_0)(\widehat{\mathbf{\gamma}}-\mathbf{\gamma}_0)+D_{\xi}\mathbf{U}(\mathbf{\beta}_0, \mathbf{\gamma}_0, \xi_0)(\widehat{\xi}-\xi_0)\}+o_{p^*}(1).
\tag{17}
\end{align*}

\end{lemma}

\begin{proof}[Proof of asymptotic normality in Theorem \ref{thm:1}]
We now verify all the conditions in Lemma \ref{lm:2}. Condition (i) holds because $\{\mathbbm{1}(t \leq \nu) U(\mathbf{\beta})+\mathbbm{1}(t > \nu) \widetilde{U}(\mathbf{\beta}, \mathbf{\gamma}, \xi): t \in \mathcal{T}, \mathbf{\beta} \in \Theta, \mathbf{\gamma} \in \Gamma, \xi \in \Xi, \rho\{(\mathbf{\gamma}, \xi),(\mathbf{\gamma}_0, \xi_0)\} \leq \varepsilon_n\}$ is Donsker, as was shown during the proof of consistency in Theorem \ref{thm:1}.

By Conditions 2 and 5, we have $E[|\mathbbm{1}(t \leq \nu) U(\mathbf{\beta_0})+\mathbbm{1}(t > \nu) \widetilde{U}(\mathbf{\beta}_0, \mathbf{\gamma}_0, \xi_0)|^2] < \infty$, which follows from a direct calculation. Therefore, Condition (ii) holds by the classical central limit theorem for independent and identically distributed $(x_i, y_i)$.

For Condition (iii), given that $\rho\{(\mathbf{\gamma}, \xi),(\mathbf{\gamma}_0, \xi_0)\}=O_{p^*}(n^{-1 / 2})$, we apply a Taylor expansion in $\mathbf{\beta}$ and $\mathbf{\gamma}$, along with the definition of the Fr\'echet derivative, to obtain
\begin{align*}
&\mathbf{U}(\widehat{\mathbf{\beta}}, \widehat{\mathbf{\gamma}}, \widehat{\xi})-\mathbf{U}(\mathbf{\beta}_0, \mathbf{\gamma}_0, \xi_0) \\
&\quad =\mathbf{U}(\widehat{\mathbf{\beta}}, \widehat{\mathbf{\gamma}}, \widehat{\xi})-\mathbf{U}(\mathbf{\beta}_0, \widehat{\mathbf{\gamma}}, \widehat{\xi})+\mathbf{U}(\mathbf{\beta}_0, \widehat{\mathbf{\gamma}}, \widehat{\xi})-\mathbf{U}(\mathbf{\beta}_0, \mathbf{\gamma}_0, \widehat{\xi})+\mathbf{U}(\mathbf{\beta}_0, \mathbf{\gamma}_0, \widehat{\xi})-\mathbf{U}(\mathbf{\beta}_0, \mathbf{\gamma}_0, \xi_0) \\
&\quad =D_{\mathbf{\beta}}\mathbf{U}(\mathbf{\beta}_0, \widehat{\mathbf{\gamma}}, \widehat{\xi})(\widehat{\mathbf{\beta}}-\mathbf{\beta}_0)+D_{\mathbf{\gamma}}\mathbf{U}(\mathbf{\beta}_0, \mathbf{\gamma}_0, \widehat{\xi})(\widehat{\mathbf{\gamma}}-\mathbf{\gamma}_0)+P\left(\mathbbm{1}(t > \nu)D_{\xi} \widetilde{U}(\mathbf{\beta}_0, \mathbf{\gamma}_0, \xi_0)\right)(\widehat{\xi}-\xi_0) \\
&\quad \quad +o(|\widehat{\mathbf{\beta}}-\mathbf{\beta}_0|)+o(|\widehat{\mathbf{\gamma}}-\mathbf{\gamma}_0|)+o(\|\widehat{\xi}-\xi_0\|).
\end{align*}
By Conditions 7 (1)-(2), and the continuous mapping theorem, we have
\[
|D_{\mathbf{\beta}}\mathbf{U}(\mathbf{\beta}_0, \widehat{\mathbf{\gamma}}, \widehat{\xi})-D_{\mathbf{\beta}}\mathbf{U}(\mathbf{\beta}_0, \mathbf{\gamma}_0, \xi_0)|=o_{p^*}(1)
\]
and
\[
|D_{\mathbf{\gamma}}\mathbf{U}(\mathbf{\beta}_0, \mathbf{\gamma}_0, \widehat{\xi})-D_{\mathbf{\gamma}}\mathbf{U}(\mathbf{\beta}_0, \mathbf{\gamma}_0, \xi_0)|=o_{p^*}(1).
\]
Finally, by Condition 8, $D_{\xi}\mathbf{U}(\mathbf{\beta}_0, \mathbf{\gamma}_0, \xi_0) = P\left(\mathbbm{1}(t > \nu)D_{\xi} \widetilde{U}(\mathbf{\beta}_0, \mathbf{\gamma}_0, \xi_0)\right)$. Thus, Condition (iii) holds.

Since both $\widehat{\mathbf{\gamma}}$ and $\widehat{\xi}$ are $n^{1 / 2}$-consistent, Condition (iv) holds automatically under Conditions (i)-(iii).

Hence, by Lemma 2, $\widehat{\mathbf{\beta}}$ is $n^{1 / 2}$-consistent and
\begin{align*}
n^{1 / 2}(\widehat{\mathbf{\beta}}-\mathbf{\beta}_0) &= -\{D_{\mathbf{\beta}}\mathbf{U}\left(\mathbf{\beta}_0, \mathbf{\gamma}_0, \xi_0\right)\}^{-1} n^{1 / 2}\{\left(\mathbf{U}_n-\mathbf{U}\right)\left(\mathbf{\beta}_0, \mathbf{\gamma}_0, \xi_0\right) \\
&\quad \quad +D_{\mathbf{\gamma}}\mathbf{U}(\mathbf{\beta}_0, \mathbf{\gamma}_0, \xi_0)(\widehat{\mathbf{\gamma}}-\mathbf{\gamma}_0)+D_{\xi}\mathbf{U}(\mathbf{\beta}_0, \mathbf{\gamma}_0, \xi_0)(\widehat{\xi}-\xi_0)\}+o_{p^*}(1).
\end{align*}
Note that $n^{1 / 2}\{(\mathbf{U}_n-\mathbf{U})(\mathbf{\beta}_0, \mathbf{\gamma}_0, \xi_0)+D_{\mathbf{\gamma}}\mathbf{U}(\mathbf{\beta}_0, \mathbf{\gamma}_0, \xi_0)(\widehat{\mathbf{\gamma}}-\mathbf{\gamma}_0)+D_{\xi}\mathbf{U}(\mathbf{\beta}_0, \mathbf{\gamma}_0, \xi_0)(\widehat{\xi}-\xi_0)\}$ is a sum of independent and identically distributed terms, and the classical central limit theorem applies. Thus,
\[
n^{1 / 2}(\widehat{\mathbf{\beta}}-\mathbf{\beta}_0) \stackrel{d}{\rightarrow} \mathcal{N}(0, \mathbf{\Sigma})
\]
 where $\boldsymbol{\Sigma}$ denotes the covariance matrix of the limiting distribution characterized by Equation (17).

\end{proof}

\newpage

\bibliographystyle{abbrv}
\bibliography{ref.bib}

\bigskip

\bigskip

\end{document}